\documentclass[11pt,a4paper]{article}
\usepackage{amsfonts,amssymb,amsmath,amsthm,cite}
\usepackage{feynmf}
\usepackage[notref,notcite]{showkeys}
\evensidemargin=\oddsidemargin
\DeclareMathOperator{\Id}{Id}              
\DeclareMathOperator{\Ker}{Ker}           
\DeclareMathOperator{\Lin}{Lin}           
\DeclareMathOperator{\Res}{Res}         
\DeclareMathOperator{\Tad}{Tad}                 
\DeclareMathOperator{\Tr}{Tr}                 

\newtheorem{assumption}{Assumption}[section]
\newtheorem{theorem}[assumption]{Theorem}
\newtheorem{corollary}[assumption]{Corollary}

\newtheorem{lemma}[assumption]{Lemma}
\newtheorem{definition}[assumption]{Definition}
\newtheorem{prop}[assumption]{Proposition}
\newtheorem{remark}[assumption]{Remark}

\newcommand{\Th}{\Theta}

\newcommand{\A}{\mathcal{A}}              
\renewcommand{\a}{\alpha}                    
\renewcommand{\b}{\beta}                    
\newcommand{\B}{\mathcal{B}}              
\newcommand{\C}{\mathbb{C}}              
\newcommand{\CC}{\mathcal{C}}              
\newcommand{\del}{\partial}                    
\newcommand{\DD}{\mathcal{D}}           
\newcommand{\eps}{\varepsilon}          
\newcommand{\Ga}{\Gamma}                
\newcommand{\ga}{\gamma}                 
\renewcommand{\H}{\mathcal{H}}          
\newcommand{\half}{{\mathchoice{\thalf}{\thalf}{\shalf}{\shalf}}}
\newcommand{\hideqed}{\renewcommand{\qed}{}} 
\renewcommand{\L}{\mathcal{L}}          
\newcommand{\la}{\lambda}                   
\newcommand{\M}{\mathcal{M}}          
\newcommand{\N}{\mathbb{N}}            
\newcommand{\om}{\omega}                  
\newcommand{\ox}{\otimes}                  
\newcommand{\pa}{\partial}
\newcommand{\R}{\mathbb{R}}             

\newcommand{\set}[1]{\{\,#1\,\}}              
\newcommand{\shalf}{{\scriptstyle\frac{1}{2}}} 
\renewcommand{\SS}{\mathcal{S}}        
\newcommand{\thalf}{\tfrac{1}{2}}            
\newcommand{\wt}{\widetilde}                 
\newcommand{\Z}{\mathbb{Z}}                 

\def\<#1,#2>{\langle#1\,,\,#2\rangle}      
\newcommand{\norm}[1]{\left\lVert#1\right\rVert}    

\newcommand{\be}{\begin{enumerate}}
\newcommand{\ee}{\end{enumerate}}

\newbox\ncintdbox \newbox\ncinttbox
\setbox0=\hbox{$-$} \setbox2=\hbox{$\displaystyle\int$}
\setbox\ncintdbox=\hbox{\rlap{\hbox
    to \wd2{\kern-.1em\box2\relax\hfil}}\box0\kern.1em}
\setbox0=\hbox{$\vcenter{\hrule width 4pt}$}
\setbox2=\hbox{$\textstyle\int$}
\setbox\ncinttbox=\hbox{\rlap{\hbox
    to \wd2{\kern-.14em\box2\relax\hfil}}\box0\kern.1em}
\newcommand{\ncint}{\mathop{\mathchoice{\copy\ncintdbox}
    {\copy\ncinttbox}{\copy\ncinttbox}
    {\copy\ncinttbox}}\nolimits}
\newcommand{\sg}{\sigma}                              

\begin{document}

\thispagestyle{empty}

\begin{center}

CENTRE DE PHYSIQUE TH\'EORIQUE$\,^1$\\
CNRS--Luminy, Case 907\\
13288 Marseille Cedex 9\\
FRANCE\\

\vspace{3cm}

{\Large\textbf{Tadpoles and commutative spectral triples}} \\
\vspace{0.5cm}

{\large  B. Iochum$^{1, 2}$, C. Levy$^{1, 2}$} \\

\vspace{1.5cm}

{\large\textbf{Abstract}}
\end{center}

\begin{quote}
Using the Chamseddine--Connes approach of the noncommutative action on spectral triples, we show that there are no tadpoles of any order for compact spin manifolds without boundary, and also consider a case of a chiral boundary condition. Using pseudodifferential techniques, we track zero terms in spectral actions.
\end{quote}

\vspace{2cm}

\noindent
PACS numbers: 11.10.Nx, 02.30.Sa, 11.15.Kc

MSC--2000 classes: 46H35, 46L52, 58B34

CPT-P005-2009

\vspace{5.5cm}

{\small
\noindent $^1$ UMR 6207

-- Unit\'e Mixte de Recherche du CNRS et des
Universit\'es Aix-Marseille I, Aix-Marseille II et de l'Universit\'e
du Sud Toulon-Var

-- Laboratoire affili\'e \`a la FRUMAM -- FR 2291\\
$^2$ Also at Universit\'e de Provence,
iochum@cpt.univ-mrs.fr, levy@cpt.univ-mrs.fr
}

\newpage

\section{Introduction}

The history of noncommutative residue is now rather long
\cite{Kassel}, so we sketch it only briefly: after some approaches by
Adler \cite{Adler} and  Manin \cite{Manin} on the Korteweg-de Vries
equation using a trace on the algebra of formal pseudodifferential
operators in one dimension, and of Guillemin with his "soft" proof of
Weyl's law on the eigenvalues of an elliptic operator
\cite{Guillemin}, the noncommutative residue  in any dimension was essentially initiated par Wodzicki in his thesis \cite{Wodzicki2}. This residue gives the unique non-trivial trace on the algebra of pseudodifferential operators. Then, a link between this residue and the Dixmier's trace was given by Connes in \cite{Connesaction}. Thanks to Connes again \cite{Book,Cgeom}, the setting of classical pseudodifferential operators on Riemannian
manifolds without boundary was extended to a noncommutative geometry
where the manifold is replaced by a non necessarily commutative
algebra $\A$ plus a Dirac-like operator $\DD$ via the notion of
spectral triple $(\A,\,\H,\,\DD)$ where $\H$ is the Hilbert space
acted upon by $\A$ and $\DD$. The previous Dixmier's trace is extended to the algebra of pseudodifferential operators naturally associated to the triple $(\A,\,\H,\,\DD)$. This spectral point of view appears quite natural in the general framework of noncommutative geometry which goes beyond Riemannian geometry. From a physicist point of view, this framework has many advantages: the spectral approach is motivated by quantum physics but not only since classical observables and infinitesimals are now on the same footing and even Dixmier's trace is related to renormalization.  It is amazing to observe that most of classical geometrical notions like those defined in relativity or particle physics can be extended in
this really noncommutative setting. Among others, some physical
actions still makes sense as in \cite{Connesaction}  where Dixmier's trace is used to compute the Yang--Mills action in the context of noncommutative differential geometry. Another example is the Einstein--Hilbert action: on a compact
spin Riemannian 4-manifold, $\ncint \DD^{-2}$ coincides (up to a
universal scalar) with the Einstein--Hilbert action, where $\ncint$ is
precisely the noncommutative residue, a point first noticed by Connes; then, there were some brute force proof \cite{Kastler} and generalization \cite{KW} (see also \cite{Ack}) of this fact which is particularly relevant here.

Since then, the case of compact manifolds with
boundary have been studied, making clearer the links between
noncommutative residues, Dixmier's trace and heat kernel expansion. This was made 
using Boutet de Monvel's algebra \cite{FGLS, Schrohe, GSc}, in the
case of conical singularities \cite{Schrohe1, Lescure} or when the
symbols are log-polyhomogeneous \cite{Lesch}. Besides, the
applications of noncommutative residues for such manifolds to classical gravity has begun
\cite{Wang}, and better, when the gravity is unified with fundamental
interactions \cite{CC2}. Needless to say that in field theory, the
one-loop calculation divergences, anomalies and different asymptotics
of the effective action are directly obtained from the heat kernel
method \cite{V}, so all of the above quoted mathematical results have
profound applications to physics. 
\medskip

The Chamseddine--Connes action \cite{CC} associated to a spectral
triple $(\A,\, \H, \,\DD)$ is, for a one-form $A=\sum_i a_i[\DD,b_i]$,
$a_i,\,b_i\in \A$
\begin{align}
   \label{formuleaction}
\SS(\DD_{A},\Phi,\Lambda) \, = \,\sum_{0<k\in Sd^+} \Phi_{k}\,
\Lambda^{k} \ncint \vert \DD_{A}\vert^{-k} + \Phi(0) \,
\zeta_{\DD_{A}}(0) +\mathcal{O}(\Lambda^{-1})
\end{align}
where $\DD_{A}:=\DD+A$ (or $\DD_{\tilde A}:=\DD +\wt A 
, \wt A:=A+\epsilon JAJ^{-1}$ in the real case),)  $\Phi_{k}=
\half\int_{0}^{\infty} \Phi(t) \, t^{k/2-1} \, dt$ and
$Sd^+$ is the strictly positive part of the dimension spectrum of the 
spectral triple. When $\DD_A$ is not invertible, we invert in \eqref{formuleaction} the invertible operator $\DD_A+P_A$ where $P_A$ is the projection on $\Ker \DD_A$ which is a finite dimensional space.

The coefficient $\zeta_{\DD_A}(0)$ related to the
constant term in (\ref{formuleaction}) can be computed from the unperturbed spectral
action since it has been proved in \cite{CC1} (with an invertible
Dirac operator and a 1-form $A$ such that $\DD+A$ is also invertible)
that
\begin{align}
   \label{constant}
\zeta_{\DD+A}(0)-\zeta_{\DD}(0)= \sum_{q=1}^{n}\tfrac{(-1)^{q}}{q}
\ncint (A\DD^{-1})^{q},
\end{align}
using $\zeta_X(s)=\Tr(|X|^{-s})$.

It is important to be able to compute \eqref{formuleaction} and here,
we look at possible cancellation of terms in this formula. We focus
essentially on commutative spectral triples, where we show that there
are no tadpoles, i.e. terms like $\ncint A\DD^{-1}$ are zero: in field
theory, $\DD^{-1}$ is the Feynman propagator and $A\DD^{-1}$ is a
one-loop graph with fermionic internal line and only one external
bosonic line $A$ looking like a tadpole. More generally, the tadpoles are the $A$-linear terms in \eqref{formuleaction}.
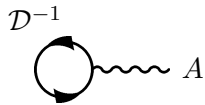
\begin{fmffile}{tadpolegraph}
\[
{\begin{picture}(100,60)
\put(0,0){\begin{fmfgraph}(50,60)
\fmfleft{l}
\fmfright{r}
\fmffreeze
\fmf{photon,tension=3}{r,w}
\fmf{photon,tension=3}{w,v}
\fmf{fermion,right}{v,l}
\fmf{fermion,right}{l,v}
\end{fmfgraph}}
\put(-10,45){\mbox{${\cal D}^{-1}$}}
\put(55,27){\mbox{$A$}}
\end{picture}} 
\]
\end{fmffile}
\vspace{-1cm}

In \cite{Ponge}, few computations of $\ncint \vert \DD \vert^{-k}$ are presented and formula like \eqref{constant} also appears in \cite{LP} in the context of pseudodifferential elliptic operators.

For examples of spectral action in the real noncommutative setting,
see \cite{CCM, Carminati, Knecht} for the case of almost commutative
cases which pops up in particle physics and \cite{GI2002} for the
Moyal plane (and few points for non compact manifolds \cite{GIV}),
\cite{GIVas, MCC} for the noncommutative torus and \cite{MC} for the
quantum group $SU_q(2)$. In this last case, there are tadpoles.

As a starting point, we investigate in section 2 the existence of tadpoles for manifolds with boundaries, considering after Chamseddine and Connes \cite{CC2} the case of a chiral boundary condition on the Dirac operator. One of their original motivations was to show that the first two terms in spectral action come with the right ratio and sign for their coefficients as in the modified Euclidean action used in gravitation. We generalize this approach to the perturbed Dirac operator by an internal fluctuation, ending up with no tadpoles up to order 5 (see definition \ref{Deftadpole}.)

However, this approach stems from explicit computations of first heat kernel coefficients, so we cannot conclude that other integrals of the same type as tadpoles are zero. It is then natural to restrict to manifolds without  boundary via a different method. 

We gather in section 3 some basic results concerning the use of the reality operator $J$. After some useful facts using the link between $\ncint$ and the Wodzicki residue, we conclude in section 4 that a lot of terms in \eqref{formuleaction} are zero, using pseudodifferential techniques. 

Few definitions about pseudodifferential operators, dimension spectrum have been postponed in the appendix.

\section{Tadpoles and compact spin manifolds with boundary}

Let $M$ be a smooth compact Riemannian $d$-dimensional manifold with
smooth boundary $\del M$ and $V$ be a given smooth vector bundle on
$M$. We denote $dx$ (resp. $dy$) the Riemannian volume form on $M$ (resp. on $\del M$.)

Recall that a differential operator $P$ is of Laplace type if it has locally the form
\begin{equation}
P = - (g^{\mu\nu} \del_\mu \del_\nu + \mathbb{A}^\mu\del_\mu +\mathbb{B})  \label{Lapl}
\end{equation}
where $(g^{\mu\nu})_{1\leq \mu,\nu\leq d}$ is the inverse matrix associated
to the metric $g$ on $M$, and $\mathbb{A}^\mu$ and $\mathbb{B}$ are smooth
$L(V)$-sections on $M$ (endomorphisms). A differential operator $D$ is of Dirac type if $D^2$ is of Laplace type, or equivalently if it has locally the following form
$$
D = -i \ga^\mu \del_\mu + \phi
$$
where $(\ga^\mu)_{1\leq \mu\leq d}$ gives $V$ a Clifford module structure: $\set{\ga^\mu,\ga^\nu}=2g^{\mu\nu}\Id_V$, ${(\ga^\mu)}^*=\ga^\mu$.

A particular case of Dirac operator is given by the following formula 
\begin{equation}
D = -i\ga^\mu (\del_\mu + \om_\mu) \label{phiDirac}
\end{equation}
where the $\om_\mu$ are in $C^\infty\big(L(V)\big)$.

If $P$ is a Laplace type operator of the form (\ref{Lapl}), then (see \cite[Lemma 1.2.1]{Gilkey2})
there is an unique connection $\nabla$ on $V$ and an unique
endomorphism $E$ such that $P = L(\nabla,E)$ where by definition
\begin{align*}
&L(\nabla,E) :=  -(\Tr_g \nabla^2  + E), \quad \nabla^2(X,Y):= [\nabla_X,\nabla_Y] -\nabla_{\nabla^{g}_X Y} \, ,
\end{align*}
$X,Y$ are vector fields on $M$ and $\nabla^g$ is the Levi-Civita connection on $M$. Locally 
$$
\Tr_g \nabla^2 := g^{\mu\nu}(\nabla_\mu \nabla_\nu -\Ga^{\rho}_{\mu\nu} \nabla_\rho)
$$
where $\Ga^{\rho}_{\mu \nu}$ are the Christoffel coefficients of $\nabla^g$.
Moreover (with local frames of $T^*M$ and $V$), $\nabla =
dx^\mu\ox (\del_\mu +\om_\mu)$ and $E$ are related to $g^{\mu\nu}$,
$\mathbb{A}^\mu$ and $\mathbb{B}$ through
\begin{align}
\om_\nu&=  \half g_{\nu\mu}(\mathbb{A}^\mu +g^{\sg\eps} \Ga_{\sg \eps}^{\mu}\Id )
\label{omeganu}\, ,\\
E&=  \mathbb{B}-g^{\nu\mu}(\del_{\nu} \om_\mu +\om_\nu\om_\mu -\om_\sg
\Ga_{\nu\mu}^\sg ) \label{EEquation}  \, .
\end{align}

Suppose that $P=L(\nabla,E)$ is a Laplace type operator on $M$, and assume that $\chi$ is an endomorphism of $V_{\del M}$ so that
$\chi^2=\Id_V$. We extend $\chi$ on a collar neighborhood $\CC$ of $\del M$
in $M$ with the condition $\nabla_d \,(\chi) = 0$ where the
$d^{th}$-coordinate here is the radial coordinate (the geodesic
distance of a point in $M$ to the boundary $\del M$.)

Let $V_\pm:=\Pi_{\pm} V$ the sub-bundles of $V$ on $\CC$ where $\Pi_\pm:=\half(\Id_V\pm \chi)$ are the projections on the $\pm1$ eigenvalues of $\chi$.
We also fix an auxiliary endomorphism $S$ on ${V_+}_{\del M}$ extended to $\CC$.

This allows to define the mixed boundary operator $\B=\B(\chi,S)$ as
\begin{equation}
\B s := \Pi_+ (\nabla_{d} +S)\Pi_+ s_{ |\del M} \oplus
\Pi_{-}s_{|\del M}\,, \quad s\in C^\infty(V) .\label{Mixed}
\end{equation}
These boundary conditions generalizes Dirichlet ($\Pi_-=\Id_V$) and Neumann--Robin ($\Pi_+=\Id_V$) conditions.

We define $P_\B$ as the realization of $P$ on $\B$, that is to say the closure of $P$ defined
on the space of smooth sections of $V$ satisfying the boundary condition $\B s=0$.
\medskip

We are interested in the behavior of heat kernel coefficients $a_{d-n}$ defined through its expansion as $\Lambda\to \infty$ (see \cite[Theorem 1.4.5]{Gilkey2})
$$
\Tr(e^{-\Lambda^{-2} D^2_\B}) \sim \sum_{n\geq 0} \Lambda^{d-n}\,
a_{d-n}(D,\B)
$$ 
where $D$ is a self-adjoint Dirac type operator. Moreover, we will use a perturbation $D\to D+A$, where $A$ is a 1-form (a linear combination of terms of the type $f[D,g]$, where $f$ and $g$ are smooth functions on $M$). More precisely, we investigate the linear
dependence of these coefficients with respect to $A$. It is clear
that, since $A$ is differential operator of order 0, a perturbation
$D\mapsto D+A$ transforms a Dirac type operator into another Dirac
type operator. 

This perturbation has consequences on the $E$ and $\nabla$ terms:

\begin{lemma}
\label{Perturbation}
Let $D$ be a Dirac type operator locally of the form
(\ref{phiDirac}) such that $\nabla_\mu:=\partial_\mu + \om_\mu$ is
connection compatible with the Clifford action $\ga$. Let $A$ be a 1-form associated to $D$, so that $A$ is locally of the form $-i\ga^\mu a_\mu$ with $a_\mu \in C^\infty(U)$, $(U,x_\mu)$ being a local coordinate frame on $M$.  

Then $(D+A)^2=L(\nabla^A,E^A)$ and $D^2=L(\nabla,E)$ where, 
\begin{align*}
&\om^A_{\mu}= \om_\mu + a_\mu\, ,\,\,  \text{thus }
\nabla^A_\mu=\nabla_\mu+a_\mu \,\text{Id}_V,\\
&E^A = E +\tfrac{1}{4}[\ga^\mu, \ga^\nu]F_{\mu \nu} , \quad E=\tfrac{1}{2}\ga^\mu \ga^\nu [\nabla_{\mu}, \nabla_\nu] , \quad F_{\mu\nu}:=\partial_\mu(a_\nu)-\partial_\nu(a_\mu)
\end{align*}
Moreover, the curvature of the connection $\nabla^A\,$is
$\Omega_{\mu \nu}^{A}=\Omega_{\mu \nu}+F_{\mu \nu}$, where $\Omega_{\mu\nu}=[\nabla_\mu,\nabla_\nu]$.

In particular $\Tr E^A=\Tr E$.
\end{lemma}

\begin{proof}
This is quoted in \cite[equation (3.27)]{V}.

$(D+A)^2=L(\nabla^A,E^A):=-g^{\mu\nu}(\nabla^A_\mu \nabla^A_\nu-\Gamma_{\mu\nu}^\rho \nabla^A_\rho)-E^A$ and we get with $\nabla^A_\mu:=\nabla_\mu+a_\mu \text{Id}_V$:
\begin{align}
-(D+A)^2&=\ga^\mu\nabla^A_\mu\ga^\nu\nabla^A_\nu=\ga^\mu[\nabla^A_\mu,\ga^\nu]\nabla^A_\nu+\ga^\mu\ga^\nu \nabla^A_\mu \nabla^A_\nu \nonumber \\
&=\ga^\mu[\nabla_\mu,\ga^\nu]\nabla^A_\nu+\tfrac{1}{2}(\ga^\mu\ga^\nu+\ga^\nu\ga^\nu)\nabla^A_\mu\nabla^A_\nu + \tfrac{1}{2}\ga^\mu \ga^\nu[\nabla^A_\mu,\nabla^A_\nu] \nonumber \\
&=-\ga^\mu\ga^\rho{\Gamma_{\mu \rho}}^\nu \nabla^A_\nu+g^{\mu
\nu}\nabla^A_\mu\nabla^A_\nu
+\tfrac{1}{2}\ga^\mu \ga^\nu[\nabla_\mu +a_\mu Id_V,\nabla_\nu +a_\nu \text{Id}_V]. \label{nabla}
\end{align}
Since $\Gamma_{\mu\nu}^\rho=\Gamma_{\nu\mu}^\rho$, we get by comparison, 
\begin{align*}
E^A&=\tfrac{1}{2}\ga^\mu \ga^\nu[\nabla_\mu +a_\mu \,\text{Id}_V,\nabla_\nu
+a_\nu \,\text{Id}_V]=\tfrac{1}{2}\ga^\mu \ga^\nu \big([\nabla_\mu ,\nabla_\nu
]+\partial_\mu(a_\nu)-\partial_\nu(a_\mu)\big) \\
&=\tfrac{1}{2}\ga^\mu \ga^\nu [\nabla_\mu ,\nabla_\nu
] +\tfrac{1}{4}[\ga^\mu, \ga^\nu] \big(\partial_\mu(a_\nu)-\partial_\nu(a_\mu)\big).
\tag*{\qed}
\end{align*}
\hideqed
\end{proof}
Remark that even if quadratic terms in $A^2$ appear in the local presentation of the perturbation $D^2\to (D+A)^2$ (in the $b$ term), these terms do not appear in the invariant formulation $(\nabla,E)$
since there are hidden in $\nabla^A_\mu\nabla^A_\nu$ of \eqref{nabla}. 

\medskip

In the following, $D$ and $A$ are fixed and satisfy the hypothesis of Lemma \ref{Perturbation}. Indices $i$, $j$, $k$, and $l$ range from 1 through
the dimension $d$ of the manifold and index a local orthonormal frame
$\{ e_1,...,e_d\}$ for the tangent bundle. Roman indices
$a$, $b$, $c$, range from 1 through $d-1$ and index a local orthonormal
frame for the tangent bundle of the boundary
$\partial M$. The vector field $e_d$ is chosen to be the inward-pointing unit normal
vector field. Greek indices are associated to coordinate frames. 

Let $R_{ijkl}$, $\rho_{ij}:=R_{ikkj}$ and $\tau:=\rho_{ii}$ be respectively the components of the Riemann tensor, Ricci tensor and scalar curvature of the Levi-Civita connection. 
Let $L_{ab}:=(\nabla_{e_{a}}e_{b},e_{d})$ be the second fundamental form of the hypersurface $\partial M$ in $M$. 
Let
``;"  denote multiple covariant differentiations with respect to $\nabla^{A}$ and 
``:" denote multiple covariant differentiations with respect to $\nabla$ and the Levi-Civita
connection of $M$.

We will look at a chiral boundary condition. This is a mixed boundary condition natural to consider in order to preserve the existence of chirality on $M$ and its boundary $\del M$ which are compatible with the (selfadjoint) Clifford action: we assume that the operator $\chi$ is selfadjoint and satisfies the following relations:
\begin{align}
\{\chi,\ga^d\} = 0 \,, \qquad  [\chi, \ga^a] = 0 \label{chigamma},\, \forall a \in \set{1,\cdots,d-1}\,.
\end{align}

This condition was shown in \cite{CC2} a natural assumption to enforce the hermiticity of the realization of the Dirac operator. It is known \cite[Lemma 1.5.3]{Gilkey2} that ellipticity is preserved.

Since $\ga^d$ is invertible, $\dim V_+ = \dim V_-$ and $\Tr\chi =0$. 

For an even-dimensional oriented manifold, there is a natural candidate $\chi$ satisfying \eqref{chigamma}, namely 
$$
\chi:={\chi_{}}_{\pa M}=(-i)^{d/2-1} \ga(e_1)\cdots \ga(e_{d-1})
$$
(this notation is compatible with \eqref{chi}.)
Recall that 
\begin{align}
\Tr ( \ga^{i_1} \cdots \ga^{i_{2k+1}} ) = 0 \, ,\,\, \forall k \in \N, \quad\Tr (\ga^{i} \ga^{j} ) = \dim V \, \delta ^{ij}\, . \label{tracegammaimpair}
\end{align}

The natural realization of this boundary condition for the Dirac type operator $D+A$ is the operator $(D+A)_\chi$ which acts as $D+A$ on the domain $\set{s\in C^\infty(V) \, : \, \Pi_- s_{|\del M} =0 }$. It turns out (see \cite[Lemma 7]{BG2}) that the natural boundary operator $B_\chi^A$ defined by 
$$
\B_\chi^A s:= \Pi_- (D+A)^2 s_{|\del M} \oplus \Pi_{-}s_{|\del M}\, 
$$  
is a boundary operator of the form (\ref{Mixed}) provided that 
$S = \half \Pi_+ (-i[\ga^d,A]-L_{aa}\chi)\Pi_+$. 

\begin{lemma}
Actually, $S$ and $\chi_{;a}$ are independent of the perturbation $A$:
\label{Schi}

(i) $S = -\half L_{aa}\, \Pi_+$\,.

(ii) $\chi_{;a} = \chi_{:a}$.
\end{lemma}

\begin{proof} 
$(i)$ Since $A$ is locally of the form $-i \ga^j a_j$ with $a_j\in C^\infty(U)$,  we obtain from (\ref{chigamma}), 
\begin{align*}
\chi[\ga^d,A] = -i a_j \, \chi[\ga^d,\ga^j] = -i \sum_{j<d} a_j \,\chi[\ga^d,\ga^j] = i\sum_{j<d} a_j \,[\ga^d,\ga^j]\chi = -[\ga^d,A] \chi
\end{align*}
and the result as a consequence of $\Pi_+\,[\ga^d,A]=[\ga^d,A]\,\Pi_-$ and $\Pi_+\Pi_-=0$.

$(ii)$ We have $\nabla^A_{i} = \nabla_i+a_i \Id_V$ where $A=:-i\ga^j a_j$, and since  
$(\nabla^A_i \chi) s = \nabla^A_i (\chi s) - \chi( \nabla^A_i s)$ 
for any $s\in C^\infty(V)$, using Lemma \ref{Perturbation}, $\nabla^A_i (\chi) =[\nabla_i + a_i \Id_V, \chi]=[\nabla_i,\chi] = \nabla_i (\chi)$.
\end{proof}

While $S$ is not sensitive to the perturbation $A$, the boundary operator $\B_\chi^A$ depends a priori on $A$. We shall denote $\B_\chi$ the boundary operator $\B_\chi^A$ when $A=0$.

The coefficients $a_{d-k}$ for $0\leq k\leq 4$ have been computed in \cite{BG1} for general mixed boundary conditions in the case of Laplace type operators and in \cite[Lemma 8]{BG2} for Dirac type operators with chiral boundary conditions. We recall here these coefficients in our setting:

\begin{prop}
\label{ThmGilkey}\begin{align*}
& a_{d}(D+A,\B_\chi^A)=(4\pi)^{-d/2}\,\int_M \Tr_V 1 \, dx \, , \\
& a_{d-1}(D+A,\B_\chi^A)=0 \, ,  \\
& a_{d-2}(D+A,\B_\chi^A)=\tfrac{(4\pi)^{-d/2}}{6}\big\{
     \int_M \Tr_V(6E^A+\tau) \,dx + \int_{\del M}\Tr_V(2L_{ aa}+12S)\, dy\,\big\}, \\
& a_{d-3}(D+A,\B_\chi^A)=\tfrac{(4 \pi )^{
       -(d-1)/2}}{384}  \int_{\del M}\Tr_V\big\{96 \chi E^A + 3 L_{aa}^2 + 6 L_{ab}^2 
+ 96 S L_{aa} + 192 S^2 -12 \chi_{;a}^2 \} \, dy ,\\
& a_{d-4}(D+A,\B_\chi^A)=\tfrac{(4 \pi )^{-d/2}}{360} \big\{\int_M \Tr_V 
      \big \{ 60\tau E^A + 180 (E^A)^2 +30 (\Omega_{ij}^{A})^2 + 5\tau^2 -2\rho^2+2R^2 \big\} \, dx \\
      &\hspace{4.3cm} + \int_{\del M} \Tr_V \big\{ 180 \chi E^A_{;d}+120E^A L_{aa} +720 S E^A+60 \chi \chi_{;a} \Omega^A_{ad} +T\big \} \, dy\, \big\}.
\end{align*}
where 
\begin{align*}
&T:=20 \tau L_{aa}+  4 R_{adad} L_{bb} -12 R_{adbd}L_{ab}+4R_{abcb}L_{ac} + \tfrac{1}{21}\big(160 L_{aa}^3-48L_{ab}^2 L_{cc}+272 L_{ab}L_{bc}L_{ac}\\
&\hspace{1cm} +120 \tau S+144 S L_{aa}^2+48 S L_{ab}^2+480 (S^2 L_{aa}+S^3)-42\chi_{;a}^2 L_{bb}+6\chi_{;a}\chi_{;b}L_{ab}-120 \chi_{;a}^2 S\big)
\end{align*}
is independent of $A$.
\end{prop}

The following proposition shows that there are no tadpoles in manifolds endowed with a chiral boundary condition.
\begin{theorem}
Let $M$ be an even $d$-dimensional compact oriented spin Riemannian manifold with smooth boundary $\del M$ and spin bundle $V$. Let $D:=-i\ga^j\nabla_j $ be the classical Dirac operator, and $\chi={\chi_{}}_{\pa M}=(-i)^{d/2-1} \ga(e_1)\cdots \ga(e_{d-1})$ where $(e_i)_{1\leq i\leq d}$ is a local orthonormal frame of $TM$.
 
The perturbation $D\to D+A$ where $A=-i\ga^j a_j$ is a 1-form for $D$, induces, under the chiral boundary condition, the following perturbations on the heat kernel coefficients where we set 
$c_{d-k}(A):=a_{d-k}(D+A,\B_\chi^A)-a_{d-k}(D,\B_{\chi})$:

(i) $c_{d}(A)=c_{d-1}(A)=c_{d-2}(A)=c_{d-3}(A) = 0$.

(ii) $c_{d-4}(A) =- \tfrac{1}{6(2\pi)^{d/2}}\int_M F_{\mu\nu}F^{\mu\nu} \,dx.$

In other words, the coefficients $a_{d-k}$ for $0\leq k\leq 3$ are unperturbed, $a_{d-4}$ is only perturbed by quadratic terms in $A$ and there are no linear terms in $A$ in $a_{d-k}(D+A,\B_\chi^A)$ for $k\leq 5$.
\end{theorem}

\begin{remark} When $A$ is selfadjoint, all coefficients $a_{d-k}(D+A,\B_\chi^A)$ and $a_{d-k}(D,\B_{\chi})$ are real while linear contributions in $A$ are purely imaginary, modulo traces of $\ga$ and $\chi$ matrices and their covariant derivatives. Since the invariant terms appearing as integrands of $\int_{M}$ and $\int_{\del M}$ in the coefficients at higher order are polynomial in $S$, $\chi$, $R$, $E^A$ and $\Omega^A$, and their covariant derivatives, one expects no linear terms in $A$ at any order.

We study more examples in \cite{avenir} with a generalization of Theorem \ref{proptadpoles}.
\end{remark}

\begin{proof}
$(i)$ The fact that $c_{d}(A)=c_{d-1}(A)=0$ follows from Proposition \ref{ThmGilkey}. 

Since by Lemma \ref{Schi}, $c_{d-2}(A) = (4\pi)^{-d/2}\int_M \Tr_V(E^A-E_A)\,  dx$, we get $c_{d-2}(A)=0$ because $\Tr_V E^A =\Tr_V E$ by Lemma \ref{Perturbation}. 
     
From Proposition \ref{ThmGilkey} and Lemma \ref{Schi}, we get
$c_{d-3}(A)= \tfrac{1}{4}(4 \pi )^{-(d-1)/2} \int_{\del M}\Tr_V\big\{ \chi (E^A-E)\big\}$.

Since $\chi(E^A-E) = (-i)^{d/2} \ga^1\cdots\ga^{d-1} [\ga^j,\ga^k]F_{jk}$, (\ref{tracegammaimpair}) yields  $\Tr_V \chi(E^A-E)=0$ because $d$ is even.

$(ii)$ Since $\Tr_V(E^A-E)=0$ and $\Tr_V \chi(E^A-E) =0$, we obtain $\Tr_V S (E^A-E)=0$ from Lemma \ref{Schi}. Thus, using Proposition \ref{ThmGilkey} and Lemma \ref{Schi},
\begin{align*}
&c_{d-4}(A)=\tfrac{(4 \pi )^{-d/2} }{360}\big\{\int_M \Tr_V
      \big \{ 180( (E^A)^2 -E^2) +30 \big((\Omega_{ij}^{A})^2-(\Omega_{ij})^2\big) \big\} \, dx \\
      &\hspace{4cm} + \int_{\del M} \Tr_V \big\{ 180 \chi( E^A_{;d}-E_{:d})+60 \chi \chi_{;a} (\Omega^A_{ad}-\Omega_{ad}) \} \, dy\, \big\}.
\end{align*}
We obtain locally $\Tr_V \big((E^A)^2-E^2\big) = \tfrac{1}{16}\Tr([\ga^\mu,\ga^\nu][\ga^\rho,\ga^\sg])F_{\mu\nu}F_{\rho\sg}$ using Lichn\'erowicz formula $E=-\tfrac{1}{4}\tau$. 
Since $\Tr_V([\ga^\mu,\ga^\nu][\ga^\rho,\ga^\sg])=4. 2^{d/2}(g^{\mu\sg}g^{\nu\rho}-g^{\mu\rho}g^{\nu\sg})$, 
\begin{align*}
\Tr_V \big((E^A)^2-E^2\big) = -\,2^{d/2-1} F_{\mu\nu}F^{\mu\nu}.
\end{align*}
$\nabla$ being the spin connection associated to the spin structure of $M$, we have $\Omega_{ij}=\tfrac{1}{4}\ga^{k}\ga^l R_{ijkl}$. So $R_{ijkl}=-R_{ijlk}$ implies $\Tr_V \Omega_{ij}=0$. Hence, with Lemma \ref{Perturbation}, 
\begin{align*}
\Tr_V \big((\Omega_{ij}^A)^2- \Omega_{ij}^2\big) =2^{d/2}  F_{ij}^2 =2^{d/2} F_{\mu\nu}F^{\mu\nu}\, .
\end{align*}
Moreover, $E^A_{;d} = [\nabla_d+a_d,E^A] = [\nabla_d,E+\tfrac{1}{4}[\ga^i,\ga^j]F_{ij}] = E_{:d} + \tfrac{1}{4}[\nabla_d,[\ga^i,\ga^j]] F_{ij}$. 

Using of $[\nabla_i,\ga^i]=\ga(\nabla_i e_j)$ and (\ref{tracegammaimpair}),
$$
\Tr_V \big( \chi( E^A_{;d}-E_{:d}) \big)= (-i)^{d/2} \tfrac{1}{2}\, F_{ij} \Tr_V \big\{ \ga^1\cdots \ga^{d-1} \big(\ga(\nabla_d e_i) \ga^j+\ga^i \ga(\nabla_d e_j)\big) \big\} =0  \, .
$$ 
It remains to check that $\Tr_V \big(\chi \chi_{:a} (\Omega^A_{ad}-\Omega_{ad})\big)=0$. 
Let $\chi_M=-i\chi \ga^d$ be the grading operator (see \eqref{chi}.) Since $\chi_{M}$ commutes with the spin connection operator $\nabla$ (see \cite[p. 396]{Polaris}),
$$
0=[\nabla_a,\chi_M]=[\nabla_a,\chi\ga^d] = \chi_{:a}\ga^d +\chi[\nabla_a,\ga^d]=\chi_{:a}\ga^d +\chi\ga(\nabla_a e_d)
$$
and thus $\chi\chi_{:a}= -\ga(\nabla_a e_d)\ga^d  = -\Ga_{ad}^j \ga^j \ga^d$, where 
$\Ga_{ad}^j=-\Ga_{aj}^{d}$ since $(e_j)$ is an orthonormal frame. 
So $\Tr_V(\chi \chi_{:a}) = -\Ga_{ad}^{j} \delta^{jd}=-\Ga_{ad}^{d}=0$. 
Finally, the result on $c_{d-4}$ follows from  Lemma \ref{Perturbation} as $\Tr_V\big(\chi \chi_{:a} (\Omega^A_{ad}-\Omega_{ad})\big)=\Tr_V(\chi \chi_{:a})F_{ad}$.

The coefficient $a_{d-5}(D+A,\B_\chi^A)$ is computed in \cite{BGKV}. One can check directly as above that the linear terms in $A$ are not present. The computation uses the fact that the trace of the following terms $\chi E^A_{;dd}$, $E^A_{;d}S$, $\chi (E^A)^2$, $E^A S^2$, $\chi_{;a}\chi_{;b} \Omega^A_{ab}$, $\chi_{;a}^2 E^A$, do not have linear terms in $A$.
\end{proof}

\medskip

In the following, we investigate the above conjecture with Connes--Chamseddine pseudodifferential calculus applied to compact spin manifolds without boundary and Riemannian spectral triples. We also see, using Wodzicki residue, how to compute some noncommutative integrals in this setting.  

\section{Notations and definitions}

Let $(\A,\DD,\H)$ be a spectral triple of dimension $d$.

We use the notation $D = \DD + P$, $P$ the projection on $\Ker
\DD$ implying invertibility of $D$. 

Let $J$ be the reality operator (if it exists) satisfying 
$$
J\DD=\epsilon \,\DD J,\quad \epsilon=\pm1
$$ 
according to the dimension: $\epsilon=+1$ when the dimension $d$ is  0, 2, 3, 4, 6, 7 mod 8 and $\epsilon=-1$ when
$d=$ 1, 5 mod 8. 

When the triple is even, we also use chirality operator
$\chi$ which is a grading on $\H$,  which commutes with $\A$, anti-commutes with $\DD$ and
also satisfies $J\chi=\epsilon'\,\chi J$ where $\epsilon'=1$ for
$d=0,4$ mod 8 and $\epsilon'=-1$ for $d=2,6$ mod 8.

Recall few definitions, see \cite{Cgeom, CM,MCC,Higson}:
\begin{definition}
A one-form $A$ is a finite sum of operators like $a_1[\DD,a_2]$
where $a_i \in \A$.

The set of one-forms is denoted by $\Omega_\DD^1(\A)$.
\end{definition}

\subsection{Noncommutative integrals}
We recall in Appendix few definitions about the algebra $\Psi(\A)$ of pseudodifferential operators, zeta functions and dimension spectrum.

A. Connes has introduced the following notation
$$
\ncint X:=\underset{s=0}{\Res} \,\Tr \big(X \vert \DD \vert^{-s}\big),
\quad X \in \Psi(\A).
$$
$\ncint$ is a trace on $\Psi(\A)$, (non necessarily positive, see Lemma \ref{scalarcurvature}.)

\subsubsection{Noncommutative integrals and real numbers}

\begin{lemma}
  \label{adjoint}
Let $(\A,\DD,\H)$ be a spectral triple and $X \in \Psi(\A)$. Then 
\begin{align*}
\ncint X^*=\overline{  \ncint X}.
\end{align*}
If the spectral triple is real, then, for $X \in \Psi(\A)$, $JXJ^{-1}
\in \Psi(\A)$ and
$$
\ncint JXJ^{-1}=\ncint X^*=\overline{  \ncint X}.
$$
\end{lemma}

\begin{proof}
The first result follows from (for $s$ large enough, so the operators
are traceable)
\begin{align*}
\Tr(X^*\vert \DD\vert^{-s})=\Tr \big((\vert
\DD\vert^{-\bar{s}})X)^*\big)=\overline{ \Tr(\vert \DD \vert^{-\bar{s}}
X)}=\overline{\Tr(X\vert \DD \vert^{-\bar{s}})}.
\end{align*}

The second result is due to the anti-linearity of $J$,
$\Tr(JYJ^{-1})=\overline{\Tr(Y)}$, and $J\vert \DD \vert=\vert \DD \vert
J$, so 
\begin{align*}
\Tr(X \vert \DD \vert^{-s})=\overline{\Tr(JX \vert \DD
\vert^{-s}J^{-1})}=\overline{\Tr(JXJ^{-1}\vert \DD \vert^{-\bar{s}})}.
\tag*{\qed}
\end{align*}
\hideqed
\end{proof}

\begin{corollary}
\label{reel}
For any one-form $A=A^*$, and for $k,\,l \in \N$,
$$
\ncint A^l \,\DD^{-k} \in \R,\quad \ncint \big(A\DD^{-1}\big)^k \in \R,
\quad \ncint A^l \,\vert \DD\vert^{-k} \in \R,\quad  \ncint \chi A^l\,\vert \DD \vert ^{-k} \in \R, 
\quad \ncint A^l\,\DD \, \vert \DD \vert^{-k} \in \R.
$$
\end{corollary}

\subsubsection{Tadpole}
In \cite{ConnesMarcolli}, is introduced the following

\begin{definition}
\label{Deftadpole}
In $(\A,\,\H,\,\DD)$, the tadpole
$Tad_{\DD+ A}(k)$ of order $k$, for  $k \in \set{d-l \, : \, l \in \N}$  is
the term linear in $A=A^{*}\in \Omega_\DD^1$, in the $\Lambda^k$ term
of 
\eqref{formuleaction} (considered as an infinite series) where
$\DD_{A}=\DD+ A$.

If moreover, the triple $(\A,\,\H,\,\DD,\,J)$ is real, the tadpole
$Tad_{\DD+\tilde A}(k)$ is the term linear in $A$, in the $\Lambda^k$ term
of \eqref{formuleaction} where $\DD_{A}=\DD+\wt A$.
\end{definition}

\begin{prop}
    \label{valeurtadpole}
Let $(\A,\,\H,\,\DD)$ be a spectral triple of dimension $d$ with simple dimension spectrum. Then 
\begin{align}
    \label{tadpolen-k}
&\Tad_{\DD+A}(d-k) =-(d-k)\ncint A \DD |\DD|^{-(d-k) -2}, \quad \forall  k\neq d,\\
    \label{tadpole0}
& \Tad_{\DD+ A}(0)=-\ncint  A \DD^{-1}.
\end{align}
Moreover, if the triple is real,  $\Tad_{\DD+\wt A}= 2\Tad_{\DD+A}$.
\end{prop}

\begin{proof}
By \cite[Lemma 4.6, Proposition 4.8]{MCC}, we have the following
formula,
for any $k\in \N$,
$$
\ncint |\DD_A|^{-(d-k)}=
\ncint |\DD|^{-(d-k)}+
\sum_{p=1}^k \sum_{r_1,\cdots, r_p =0}^{k-p}
\underset{s=d-k}{\Res} \, h(s,r,p) \, \Tr\big(\eps^{r_1}(Y)
\cdots\eps^{r_p}(Y) |\DD|^{-s}\big),
$$
where
\begin{align*}
\hspace{1cm}&h(s,r,p):=(-s/2)^p\int_{0\leq t_1\leq \cdots \leq
t_p\leq 1}
g(-st_1,r_1)\cdots
g(-st_p,r_p) \, dt,  \\
&\eps^r (T):=\nabla(T)\DD^{-2r}, \, \nabla(T):=[\DD^{2},T],\\
&g(z,r):=\tbinom{z/2}{r} \text{ with }g(z,0):=1,\\
&Y \sim \sum_{q=1}^N\sum_{k_1,\cdots,k_q =0}^{N-q} \Ga_q^k(X)
\DD^{-2(|k|_1+q)} \mod OP^{-N-1} \text{ for any N}\in \N^*, \\
& X:=\wt A \DD + \DD \wt A+\wt A^2, \wt A :=A+\epsilon JAJ^{-1},\\
& \Ga_q^k(X):=\tfrac{(-1)^{|k|_1+q+1}}{|k|_1+q} \nabla^{k_q}
\big(X\nabla^{k_{q-1}}(\cdots X\nabla^{k_1}(X)\cdots)\big)
\,,\,\forall q\in \N^* \,,\,
k=(k_1,\cdots,k_q)\in \N^q.
\end{align*}
As a consequence, for $k\neq n$, only the terms with $p=1$ contribute
to the linear part:
$$
\Tad_{\DD + \wt A}(d-k)= \Lin_A(\ncint |\DD_A|^{-(d-k)}) =\sum_{r=0}^{k-1}
\underset{s=d-k}{\Res} \, h(s,r,1) \,
\Tr\big(\eps^{r}(\Lin_A(Y))|\DD|^{-s}\big)\, .
$$
We check that for any $N\in \N^*$, 
$$
\Lin_A(Y) \sim \sum_{l=0}^{N-1} \Ga_1^l(\wt A \DD + \DD \wt A)
\DD^{-2(l+1)} \mod OP^{-N-1}.
$$
Since $\Ga_1^l(\wt A \DD + \DD \wt A) = \tfrac{(-1)^{l}}{l+1}
\nabla^{l}
(\wt A \DD + \DD \wt A)= \tfrac{(-1)^{l}}{l+1} \{ \nabla^l(\wt A),\DD
\}$, we get, assuming the dimension spectrum to be simple
\begin{align*}
\Tad_{\DD+ \wt A}(d-k)&= \sum_{r=0}^{k-1}
\underset{s=d-k}{\Res} \, h(s,r,p) \,
\Tr\big(\eps^{r}(\Lin_A(Y))|\DD|^{-s}\big) \\
&= \sum_{r=0}^{k-1} h(n-k,r,1) \sum_{l=0}^{k-1-r}\tfrac{(-1)^{l}}{l+1}
\underset{s=d-k}{\Res} \,\Tr\big( \eps^{r}(\{ \nabla^l(\wt A),\DD\})|\DD|^{-s-2(l+1)}\big)\\
 & =  2\sum_{r=0}^{k-1} h(d-k,r,1) \sum_{l=0}^{k-1-r}\tfrac{(-1)^{l}}{l+1}
\ncint \nabla^{r+l}(\wt A) \DD |\DD|^{-(d-k + 2(r+l)) -2}  \\ 
& = -(n-k)
\ncint \wt A \DD |\DD|^{-(d-k) -2},
\end{align*}
because in the last sum it remains only the case $r+l=0$, so $r=l=0$.

Formula \eqref{tadpole0} is a direct application of \cite[Lemma
4.5]{MCC}.

The link between  $\Tad_{\DD+\wt A}$ and $\Tad_{\DD+A}$ follows from $J\DD=\epsilon 
\DD J$ and Lemma \ref{adjoint}.
\end{proof}

\begin{corollary} 
\label{Atilde=0}
In a real spectral triple $(\A,H,\DD)$, if $A=A^*\in \Omega_\DD^1(\A)$ is such that $\wt A=0$,  then $\Tad_{D+A}(k) =0$ for any $k\in \Z$, $k\leq d$. 
\end{corollary}

\begin{remark}

Note that $\tilde A=0$ for all $A=A^{*}\in \Omega_\DD^1$, 
when $\A$ is commutative and $JaJ^{-1}=a^*$, for all $a \in \A$, see \eqref{JAJ}, so one can only use $\DD_A=\DD+A$.

But we can have $\A$ commutative and $JaJ^{-1}\neq a^*$ \cite{CGravity,
Kraj}:\\
Let $\A_1=\C \oplus \C$ represented on $\H_1=\C^3$ with, for some
complex number $m\neq0$, 
\begin{align*} 
\pi_1(a)&:=  \left( \begin{array}{ccc}
b_1 & 0 & 0\\
0 & b_1 & 0\\
0 & 0 & b_2
\end{array} \right), \,\,for \,\, a=(b_1,\,b_2) \in \A \\
\DD_1&:= \left( \begin{array}{ccc}
0 & m & m\\
\bar m & 0 & 0\\
\bar m & 0 & b
\end{array} \right),  \,\,
\chi_1:=  \left( \begin{array}{ccc}
1 & 0 & 0\\
0 & -1 & 0\\
0 & 0 & -1
\end{array} \right),  
\,\,
J_1:=\left( \begin{array}{ccc}
1 & 0 & 0\\
0 & 0 & 1\\
0 & 1 & 0
\end{array} \right) \circ \, cc
\end{align*}
where $cc$ is the complex conjugation. Then ($\A_1,\,\H_1,\,\DD_1$)
is a commutative real spectral triple of dimension $d=0$ with non
zero one-forms and such that $J_1\pi_1(a)J_1^{-1}= \pi_1(a^*)$ only if $a=(b_1,b_1)$. 

Take a commutative geometry 
\big($\A_2=C^{\infty}(M), \,\H=L^2(M,S),\,\DD_2,\, \chi_2,\, J_2$\big)
defined in \ref{rieman} where $d=dim M$ is even, and then the tensor
product of the two spectral triples, namely $\A=\A_1\otimes \A_2$,
$\H=\H_1 \otimes \H_2$, $\DD=\DD_1 \otimes \chi_2 + 1\otimes \DD_2$,
$\chi=\chi_1 \otimes \chi_2$ and $J$ is either $\chi_1 J_1\otimes
J_2$ when $d\in \set{2,6}$ mod 8 or $J_1 \otimes J_2$ in other cases,
see \cite{CGravity,Vanhecke}.

Then $(\A,\, \H,\, \DD)$ is a real commutative triple of dimension $d$ 
such that $\tilde A \neq 0$ for some selfadjoint
one-forms $A$, so is not exactly like in definition \ref{rieman}.
\end{remark}

The vanishing tadpole of order 0 has the following equivalence (see \cite{CC1})
\begin{align}
  \label{equ}
\ncint A \DD^{-1}=0, \,\forall A \in \Omega_D^1(\A)
\,\Longleftrightarrow  \,\ncint a b=\ncint a \a(b), \,\forall a,b
\in \A, 
\end{align}
where $\a(b):=\DD b\DD^{-1}$, equivalence which can be generalized as

 \begin{lemma}
\label{termecumulŽ}
In a spectral triple $(\A,\,\H,\,\DD)$, for any $k\in \N$,
\begin{align*}
& \ncint (A\DD^{-1})^n=0, \,\forall A \in \Omega^{1}_{\DD}(\A)
,\,\,\forall
n\in \set{1,\cdots,k} \Longleftrightarrow  \,\ncint \prod_{j=1}^k a_j\a(b_j)
 =\ncint \prod_{j=1}^k a_jb_j, \,\, \forall
a_j,\,b_j\in \A.
\end{align*}
\end{lemma}
 
 \begin{proof}
Note that $a[\DD,b] \DD^{-1}=a\,\tilde \a (b)$ where $\tilde \a (b)
:=\a(b)-b$.

Assuming the left hand-side, we get
\begin{align*}
0&=\ncint (A\DD^{-1})^n=\ncint a_1\tilde \a(b_1)\ldots a_j\tilde \a(b_j) \ldots a_n\tilde
\a(b_n)\\
&=\ncint a_1\tilde \a(b_1)\ldots a_j \a(b_j)a_{j+1}\tilde\a(b_{j+1})
\ldots a_k\tilde \a(b_k)-\ncint a_1\tilde \a(b_1)\ldots a_jb_ja_{j+1}\tilde \a(b_j) \ldots
a_n\tilde \a(b_n)
\end{align*}
$\forall \,a_j,\,b_j \in \A$. But the last term is zero if $\ncint
{(A\DD^{-1})}^{n-1}=0$ for all $A$. By induction, we end up with 
$0=\ncint a_1\a(b_1) \cdots a_{n-1}\a(b_{n-1})\,a_n \tilde \a(b_n)$. Varying $n$ between $1$ and $k$, we get the right hand-side.
\end{proof}

\section{Commutative spectral triples}

\subsection {Commutative geometry}

\begin{definition}
\label{rieman}
Consider a commutative spectral triple given by a compact
Riemannian spin manifold  $M$ of dimension $d$ without boundary and
its Dirac operator $\DD$ associated to the Levi--Civita connection.
This means $\big( \A:=C^{\infty}(M),\, \H:=L^2(M,S),\,\DD \big)$
where $S$
is the spinor bundle over $M$. This triple is real since, due to the
existence of a spin structure, the charge conjugation operator
generates an anti-linear isometry $J$ on $\H$ such that
$$
JaJ^{-1}=a^*,\quad \forall a \in \A,
$$
and when $d$ is even, the grading is given by the chirality matrix
\begin{align}
\label{chi}
\chi_M:=(-i)^{d/2}\,\ga^1\ga^2\cdots\ga^d.
\end{align}

Such triple is said to be a commutative geometry (see \cite{CGravity} and \cite{CReconstruction} for the role of $J$ in the nuance between spin and spin$^c$ manifold.)
\end{definition}
Since, $JaJ^{-1}=a^*$ for $a\in\A$, we get that in a commutative geometry, 
\begin{align}
\label{JAJ}
JAJ^{-1}=-\epsilon\,A^* , \quad \forall A\in \Omega^1_\DD(\A).
\end{align}

\subsection{No tadpoles}

The appearance of tadpoles never occur in commutative geometries, as
quoted in \cite[Lemma 1.145]{ConnesMarcolli} for the dimension $d=4$. 
This fact means that a given geometry $(\A,\, \H,\,\DD)$ is a
critical point for the spectral action \eqref{formuleaction}.

\begin{theorem} 
\label{proptadpoles} 
There are no tadpoles on a commutative geometry, namely, for any one-form $A=A^*\in \Omega_{\DD}^1(\A)$, $Tad_{\DD+A}(k)=0$, for any $k \in \Z$,
$k\leq d$.
\end{theorem}

\begin{proof}
Since $\tilde A=0$ when $A=A^*$ by \eqref{JAJ}, the result follows from Corollary \ref{Atilde=0}.
\end{proof}

There are similar results in the following 

\begin{lemma}
\label{similar}
Under same hypothesis, for any $k,\,l\in \N$

(i) $\ncint A\,\DD^{-k}=-\epsilon^{k+1}\,\ncint A\,\DD^{-k}$,

(ii) $\ncint \chi A\,\DD^{-k}=-\epsilon^{k+1}\ncint\chi A \,\DD^{-k}$,

(iii) $\ncint A^l \vert \DD \vert^{-k}=(-\epsilon)^l\,\ncint {A^l}\vert \DD \vert^{-k}$,

(iv) $\ncint \chi A^l \vert \DD \vert^{-k}=(-\epsilon)^l\,\ncint {\chi A^l}\vert \DD \vert^{-k}$.
\end{lemma}

\begin{proof}
\begin{align*}
\ncint A\,\DD^{-k}&=\overline{\ncint JA\,\DD^{-k}J^{-1}}=\overline{\ncint
JAJ^{-1}( \epsilon^k \DD^{-k})}=-\epsilon^{k+1}\,\overline{\ncint
A^*\,\DD^{-k}}=-\epsilon^{k+1}\,\ncint \,\DD^{-k}A \\
&=-\epsilon^{k+1}\,\ncint A\,\DD^{-k}. 
\end{align*}
The same argument gives the other equalities using $\chi A=-A \chi$ and $\chi \vert \DD \vert=\vert \DD \vert \chi$.
\end{proof}

\begin{lemma}
   \label{componentofzeta}
For any one-form $A$,
$\ncint \big( A \,\DD^{-1}\big)^k=0$ when $k\in \N$ is odd.
\end{lemma}

\begin{proof}
We have
\begin{align}
\label{trick}
\ncint \big( A \,\DD^{-1}\big)^k&=\overline{ \ncint J\big(
A\DD^{-1}\big)^kJ^{-1}}=\overline{ \ncint \big( JAJ^{-1}\,
J\DD^{-1}J^{-1}\big)^k}=(-1)^k\epsilon^{2k}\overline{ \ncint
\big(A^*\DD^{-1}\big)^k} \nonumber  \\
&=(-1)^k\ncint (A\DD^{-1})^k
\end{align}
(which shows again that $\ncint A \DD^{-1}=0$.)
\end{proof}

\subsection{Miscellaneous for commutative geometries}

To show that more noncommutative integrals, where the use of the operator $J$ in the trick \eqref{trick} is not sufficient, are nevertheless zero, we need to use the Wodzicki residue (see \cite{Wodzicki1,Wodzicki}):
in a chosen coordinate system and local trivialization $(x,\xi)$ of
$T^*M$, this residue is
\begin{align}
\label{wres}
wres_{x}(X) & :=\int_{S_x^*M} \Tr\big(\sigma_{-d}^X\,(x,\xi)\big)\,
\vert d\xi \vert\,\vert dx^1\wedge\cdots \wedge dx^d\vert,
\end{align}
where $\sigma_{-d}^X\,(x,\xi)$ is the symbol of the classical
pseudodifferential operator $X$ in the chosen coordinate frame $(x_1,\cdots,x_d)$, which is homogeneous of degree
$-d:=-\text{dim}(M)$ and taken at point $(x,\,
\xi)\in T^*(M)$, $d\xi$ is the normalized restriction of the volume
form to the unit sphere $S_x^*M \simeq \mathbb{S}^{d-1}$, so we
assume $d\geq 2$ to get $S_x^*M$ connected.

This $wres_x(X)$ appears to be a one-density not depending on the
local
representation of the symbol (see \cite{Wodzicki,Polaris}), so 
\begin{align}
Wres(X):=\int_M wres_x(X)
\end{align}
is well defined. 

The noncommutative integral $\ncint$ coincides with the Wodzicki
residue, up to a scalar: since both $\ncint$ and $Wres$ are
traces on the set of pseudodifferential operators, the uniqueness of
the trace \cite{Wodzicki} gives the proportionality  
\begin{align}
\label{Wres}
\ncint X=c_d\,Wres(X)
\end{align}
where $c_d$ is a constant depending only on $d$. Computing separately $\ncint \vert \DD \vert^{-d}$ and $Wres( \vert \DD \vert^{-d})$, we get $c_d>0$ (note that $\ncint$ is not a positive functional, see Lemma \ref{scalarcurvature}.)

Lemma \ref{adjoint} follows for instance from the fact that $\int_M
wres_x(X^*)=\overline{\int_M wres_x(X)}$.

Note that $Wres$ is independent of the metric.

As noticed by Wodzicki, $\ncint X$ is equal to $-2$ times the
coefficient in log $t$ of the asymptotics of $\Tr(X\, e^{-t\,\DD^2}$)
as $t \rightarrow 0$. It is remarkable that this coefficient is
independent of $\DD$ and this gives a close relation between the
$\zeta$ function and heat kernel expansion with $Wres$. Actually, by
\cite[Theorem 2.7]{GS}
\begin{align}
\label{heat}
\Tr(X\, e^{-t\,\DD^2}) \sim_{t\rightarrow 0^+} \sum_{k=0}^{\infty}
a_k\,t^{(j-ord(X)-d)/2} + \sum_{k=0}^{\infty} (-a'_k\,\log
t+b_k)\,t^k,
\end{align}
so $\ncint X=2 a'_0$. Since, via Mellin transform, $\Tr(X\,\DD^{-2s})
=\tfrac{1}{\Gamma (s)} \int_0^\infty t^{s-1}\, \Tr(X\, e^{-t\,\DD^2})
\,dt$, the non-zero coefficient $a'_k$, $k\neq0$ create a pole of
$\Tr(X\,\DD^{-2s}) $ of order $k+2$ since $\int_0^1 t^{s-1}
\log(t)^k=\tfrac{(-1)^k k!}{s^{k+1}}$ and 
\begin{align}
\label{Gamma}
\Gamma(s)=\frac{1}{s} +\gamma+s\,g(s)
\end{align} 
where $\gamma$ is the Euler constant and the function $g$ is also holomorphic around zero.

We have $\ncint 1=0$ and more generally, $Wres(P)=0$ for all zero-order pseudodifferential projections \cite{Wodzicki1}.

For extension to log-polyhomogeneous pseudodifferential operators,
see \cite{Lesch}.

When $M$ has a boundary, some $a'_k$ are non zero, the dimension
spectrum can be non simple (even if it is simple for the Dirac operator, see for instance \cite{Lescure}.)

On a spectral triple $(\A,\, \H,\, \DD)$, the fact to change the
product on $\A$ may or not affect the dimension spectrum: for
instance, there is no change  when one goes from the commutative
torus to the noncommutative one (see \cite{MCC}), while the dimension
spectrum of $SU_q(2$) which is bounded from below, does not coincide
with the dimension spectrum of the sphere $\mathbb{S}^3$
corresponding to $q=1$ \cite[Corollary 4.10]{MC}.

\vspace{0.3cm}
We first introduce few necessary notations. In the following we fix a local coordinate frame $(U,(x_i)_{1\leq i \leq n})$ which is normal at $x_0\in M$, and denote $\sigma_{k}^X$ the $k$-homogeneous symbol of any classical pseudodifferential operator $X$ on $M$, in this local coordinate frame. The Dirac operator is locally of the form---compatible with \eqref{phiDirac}
\begin{align}
\label{Dirac}
\DD=-i\ga (dx^j)\,\big(\partial_{x^j}+\omega_j(x)\big)
\end{align}
where $\omega_j$ is the spin connection, $\ga$ is the Clifford
multiplication of one-forms \cite[page 392]{Polaris}. Here we make the choice of gauge given
by $h:=\sqrt{g}$ which gives \cite[Exercise 9.6]{Polaris}
$$
\omega_i=-\tfrac{1}{4}\, \big( \Gamma_{ij}^k
\,g_{kl}-\partial_{x^j}(h_j^\alpha) \delta_{\a \beta}\,h_l^\beta
\big)\, \ga(dx^j)\,\ga(dx^l),\quad
\ga(dx^j)=\sqrt{g^{-1}}^{\,jk}\,\gamma_k
$$
where $\gamma^j=\gamma_j$ are the selfadjoint constant $\gamma$
matrices satisfying $\set{\ga^i ,\ga^j}=\delta^{ij}$. Thus
$$
\sigma^{\DD}(x,\xi)=\sqrt{g^{-1}}^{\,jk}\,\gamma_k \big(\xi_j-i\,\omega_j(x)\big).
$$

We have chosen normal (or geodesic) coordinates around the base point $x_0$. Since 
\begin{align*}
& g_{ij}(x)=g_{ij}(x_0)+\tfrac{1}{3}R_{ijkl}\,x^kx^l +o(\vert \vert x
\vert\vert^3),\\
& g^{ij}(x)=g^{ij}(x_0)-\tfrac{1}{3}{{{{R}^i}_k}{}^j}_l\,x^kx^l +o(\vert \vert x \vert\vert^3),\\
& g_{ij}(x_0)=\delta_{ij}, \quad\Gamma_{ij}^k(x_0)=0,
\end{align*}
the matrices $h(x)$ and $h^{-1}(x)$ have no linear terms in $x$. Thus
$$
\om_i(x_0)=0.
$$
We could also have said that parallel translation of a basis of the
cotangent bundle along the radial geodesics emanating from $x_0$
yields a trivialization (this is the radial gauge) such that
$\om_i(x_0)=0$.
In particular, using product formula for symbols and the fact that in
the decomposition $D=\DD+P$, $P\in OP^{-\infty}$, we get for $k \in
\N$
\begin{align}
&\sigma_1^{\DD}(x,\xi)=\sqrt{g^{-1}}^{\,jk}(x)\,
\gamma_k\xi_j=\ga(\xi),
&&\sigma_1^{\DD}(x_0,\xi)=\gamma^j\xi_j,\label{sigma1}\\
&\sigma_0^{\DD}(x,\xi)=-i\sqrt{g^{-1}}^{\,jk}(x)\,\gamma_k\om_j(x),
&&\sigma_0^{\DD}(x_0,\xi)=0,\label{sigma0}\\
&\partial_{x^k}\sigma_1^{\DD}(x_0,\xi)=0, \label{deriveedesigma1}\\
&\sigma_{-1}^{\DD^{-1}}(x,\xi)=\sqrt{g^{-1}}^{\,jk}(x)\,\gamma_{j}\xi_k \, \vert
\vert \xi \vert \vert_x^{-2}, && \vert \vert \xi \vert \vert_x^2:=g^{jk}(x)\,\xi_j \xi_k \label{sigma-1}\\
&\partial_{x^k}\sigma_{-1}^{\DD^{-1}}(x_0,\xi)=0. &&
\label{deriveesigma-1}
\end{align}
We will use freely the fact that the symbol of a one-form $A$ can be written as 
\begin{align}
\label{sigmaA}
\sigma^A(x,\xi)=\sigma_0^A(x)=-i \, a_k(x) \, \gamma^k
\end{align}
with $a_k(x) \in i\R$ when $A=A^*$.

When $d$ is even (so $\epsilon=1$), remark that for $k=l$ and
$A_i=a_i[\DD,b_i]$ and $a=\prod_{i=1}^k a_i$, then by \cite[page 231
(actually, $\chi$ is missing)]{CM}, \cite{Ponge2} or \cite[p.
479]{Polaris} when $k=d$, ($M$ is supposed to be oriented)
$$
\ncint \chi A_1\cdots A_k \vert \DD\vert^{-k}=c'_k \int_M
\hat{A}(R)^{(d-k)}\wedge a db_1 \wedge \cdots \wedge db_k
$$
where $\hat{A}(R)$ is the $\hat{A}$-genus associated to the Riemannian
curvature $R$. Since  we have $\hat{A}(R)\in \oplus_{j\in \N}\Omega^{4j}(M,\R)$,
$\ncint \chi A^k \vert D\vert^{-k}$ can be non zero only when
$k=d-4j$. For instance in dimension $d=$2, for $j=0$,
$$
\sigma_{-2}^{\chi A_1A_2 \DD^{-2}}(x,\xi)=\sigma_0^{\chi
A_1A_2}(x)\,\sigma_{-2}^{\DD^{-2}}(x,\xi)=-a_1(x)\,a_2(x) \,\chi
g^{jk}(x)\gamma_{j}\gamma_{k}\,\tfrac{1}{g^{lm}(x)\xi_{l}\xi_{m}}.
$$
Thus $wres_{x}(\chi A_1A_2 \DD^{-2})=-2\,a_1(x)\,a_2(x)\, \sqrt{\text{det}\, g_x} \, \Tr(\chi \gamma^{j}\gamma^{k})$,  so if $ \nu_g$ is the Riemannian density, 
\begin{align}
\label{d=2}
\ncint \chi A_1A_2\,\DD^{-2}=-2 c_d\,\, \Tr(\chi \gamma^{j}\gamma^{k})\int_M\,a_1a_2\, \nu_g .
\end{align}
Actually, this last equality is nothing else than Wodzicki--Connes' trace theorem, see \cite[section
7.6]{Polaris}, and this is equal to $c'_d\int_M   a_1a_2db_1\wedge db_2$ as claimed above.

\vspace{0.3cm}
We introduce a few subspaces of the pseudodifferential operators space
$\Psi (M)$. Let
\begin{align*}
\B_e&:=\set{P \in \Psi(M) \,:\, \sigma_j^P \in E_j, \, \forall j\in
\Z}\quad \text{e for even}, \\
\B_o & :=\set{P \in \Psi(M) \,:\, \sigma_j^P \in O_j, \, \forall j\in
\Z}\quad \text{o for odd},
\end{align*}
such that, for $m=2^{[d/2]}$,

\begin{align*}
&E_j := \set{ f\in C^{\infty} \big(U\times
\R^d\backslash\set{0},\M_m(\C) \big) \ : \   f(x,\xi) = \sum_{i\in I}
\tfrac{\xi^{\b^i}}{\norm{\xi}_x^{2k_i}} h_i(x) \ , \ I\neq \emptyset,
\\
&\hskip3cm k_i\in \N, \ \b^i \in \N^d  \ , |\b^i|-2k_i = j \ , \
h_i\in C^{\infty}(U,\M_m(\C))}\, , \\
& O_j := \set{ f\in C^{\infty} \big(U\times
\R^d\backslash\set{0},\M_m(\C) \big) \ : \   f(x,\xi) = \sum_{i\in I}
\tfrac{\xi^{\b^i}}{\norm{\xi}_x^{2k_i+1}} h_i(x) \ , \ I\neq
\emptyset, \\
&\hskip3cm k_i\in \N, \ \b^i \in \N^d,  \  |\b^i|-(2k_i+1) = j \ , \
h_i\in C^{\infty}(U,\M_m(\C))}\, .
\end{align*}

\begin{lemma} 
\label{Ejlem}
For any $j,\,j' \in \Z$ and  $\a\in \N^d$, 

$(i)$ $E_j E_{j'}\subseteq E_{j+j'} \text{ and } \del_\xi^{\a} E_j
\subseteq E_{j-|\a|},\,\,\del_x^\a E_j\subseteq E_j$.

$(ii)$ $O_j O_{j'}\subseteq E_{j+j'}$ and $\del_\xi^{\a} O_j
\subseteq O_{j-|\a|}$, $\del_x^\a O_j\subseteq O_j$.

$(iii)$ $O_j E_{j'}$ and $E_{j'} O_j$ are included in $O_{j+j'}$. 

$(iv)$ $\B_e$ is a sub-algebra of $\Psi(M)$.

$(v)$  $\B_e\B_e$, $\B_o \B_o$ are included in $\B_e$, and $\B_e
\B_o$, $\B_o \B_e$ are included in $\B_o$.
\end{lemma}
\begin{proof}
$(i)$ Let $f \in E_j$ and $\a \in \N^d$. We have, if $f(x,\xi) =
\sum_{i\in I} \tfrac{\xi^{\b^i}}{\norm{\xi}_x^{2k_i}} \,h_i(x)$, 
$$
\del^\a_\xi f = \sum_{i\in I}
\del^\a_{\xi}(\tfrac{\xi^{\b^i}}{\norm{\xi}_x^{2k_i}}) \,h_i(x)=
\sum_{i\in I} \sum_{\ga\leq \a} \tbinom{\a}{\ga}
\del^{\a-\ga}_{\xi}(\xi^{\b^i})
\del^\ga_\xi(\tfrac{1}{\norm{\xi}_x^{2k_i}}) \,h_i(x).
$$
We check by induction that we can write
$$
\del^\ga_\xi(\tfrac{1}{\norm{\xi}_x^{2k_i}}) =
\tfrac{1}{\norm{\xi}_x^{2k_i(|\ga|+1)}} \sum_{p} \la_{p}
\prod_{j=1}^{|\ga|} \del^{\b^{j,p}}_\xi \norm{\xi}_x^{2k_i}
$$
where $\la_{p}$ are real numbers, the sum on indices $p$ is
finite, and $\sum_{j=1}^{|\ga|} \b^{j,p} = \ga$. As a consequence,
since $\norm{\xi}_x^{2k_i}=(g^{kl}(x)\xi_k\xi_l)^{k_i}$ is a
homogeneous polynomial in $\xi$ of degree $2k_i$, we get 
$\del^\a_\xi f \in E_{j-|\a|}$. The inclusions $E_j E_{j'}\subseteq
E_{j+j'}$, $\del_x^\a E_j\subseteq E_j$ are straightforward.

$(ii)$ The proof is similar to $(i)$ since by induction
$$
\del^\ga_\xi(\tfrac{1}{\norm{\xi}_x}) =
\tfrac{1}{\norm{\xi}_x^{2|\ga|+1}} \sum_{p} \la_{p}
\prod_{j=1}^{|\ga|} \del^{\b^{j,p}}_\xi \norm{\xi}_x^{2}
$$
where $\la_{p}$ are real numbers, the sum on the indices $p$ is
finite and $\sum_{j=1}^{|\ga|} \b^{j,p} = \ga$.

$(iii)$ Straightforward.

$(iv)$ The product symbol formula for two classical
pseudodifferential operators $P\in \Psi^p(M)$, $Q\in \Psi^q(M)$ gives
\begin{equation}
\sg^{PQ}_{p+q-j} = \sum_{\a\in \N^d}\, \sum_{k\geq 0, \ |\a|+k\leq j}
i^{|\a|}\tfrac{(-1)^{|\a|}}{\a!} \, \del_{\xi}^\a
\sg^P_{p-j+|\a|+k}\, \del^\a_x \sg^Q_{q-k} \, .
\label{prodsymbol}
\end{equation}
The presence of the factor $i^{|\a|}$ that will be crucial in later
arguments like Lemma \ref{Calg}.

If $P,Q \in \B_e$, we see that by $(i)$, $\del_{\xi}^\a
\sg^{P}_{p-j+|\a|+k}\in E_{p-j+k}$ and $\del^\a_x \sg^Q_{q-k} \in
E_{q-k}$. 
Again by $(i)$, we obtain $\del_{\xi}^\a \sg^P_{p-j+|\a|+k}\,
\del^\a_x \sg^Q_{q-k} \in E_{p+q-j}$, so the result follows from
(\ref{prodsymbol}).

$(v)$ A similar argument as $(iv)$ can be applied, using $(ii)$ to
obtain $\B_o\B_o\subseteq \B_e$ and $(iii)$ to get $\B_o\B_e
\subseteq \B_o$, $\B_e\B_o\subseteq \B_o$.
\end{proof}

$\B_e$ and $\B_o$ are stable by inverse:

\begin{lemma}
\label{Binversion}
Let $P\in \B_e$ (resp. $\B_o$) be an elliptic classical
pseudodifferential operator in $\Psi^p(M)$ with
$\sg^P_p(x,\xi)=\norm{\xi}_x^p$, $p\in \N$. Then any parametrix
$P^{-1}$ of $P$ is in $\B_e$ (resp. $\B_o$).
\end{lemma}
\begin{proof} 
Assume $P \in \B_e$ so $p$ is even. From the parametrix equation $P
P^{-1} = 1$, we obtain  $\sg_{-p}^{P^{-1}} = (\sg_p^{P})^{-1}=
\norm{\xi}_x^{-p} \in E_{-p}$. Moreover, using (\ref{prodsymbol}), we
see that for any $j\in \N^*$,
\begin{align}
\sg_{-p-j}^{P^{-1}} = -(\sg_p^{P})^{-1}\big( \sum_{0\leq k<j}
\sg_{p-j+k}^P \, \sg_{-p-k}^{P^{-1}} 
+ \sum_{0<|\a|\leq j} \sum_{k=0}^{j-|\a|}
i^{|\a|}\tfrac{(-1)^{|\a|}}{\a!}  \,\del_{\xi}^\a \sg_{p-j+|\a|+k}^P
\, \del^\a_x \sg_{-p-k}^{P^{-1}} \,\big)
\label{parametrix}
\end{align}
We prove by induction that for any $j\in \N$, $\sg_{-p-j}^{P^{-1}}
\in E_{-p-j}$: suppose that for a $j\in \N^*$, we have for any
$j'<j$, $\sg_{-p-j'}^{P^{-1}} \in E_{-p-j'}$. We then directly check
with Lemma \ref{Ejlem} and (\ref{parametrix}) that
$\sg_{-p-j}^{P^{-1}} \in E_{-p-j}$.

The case $P\in \B_o$ is similar.
\end{proof}

\begin{lemma}
\label{DiracB}
For any $k\in \Z$, $\DD^{k} \in \B_e$ and when $k$ is odd, $|\DD|^k\in
\B_o$.
\end{lemma}

\begin{proof}
Since $\DD\in \B_e$, $\DD^{-2}$ is in $\B_e$ by Lemma
\ref{Binversion} and \ref{Ejlem} and so is $\DD^k$.

Using \eqref{prodsymbol} for the equation $|\DD||\DD|=\DD^2$, we
check that $\sg_1^{|\DD|}(x,\xi)=\norm{\xi}_x$ and for any $j\in
\N^*$,
\begin{align}
\sg_{1-j}^{|\DD|}=&\tfrac{1}{2\norm{\xi}_x}
\big(\,\sg_{2-j}^{\DD^2}-\sum_{0<k<j}\sg_{1-j+k}^{|\DD|} \,
\sg_{1-k}^{|\DD|} 
+ \sum_{0<|\a|\leq
j}\sum_{k=0}^{j-|\a|}i^{|\a|}\tfrac{(-1)^{|\a|}}{\a!} \, \del^\a_\xi
\sg_{1-j+|\a|+k}^{|\DD|} \del_{x}^\a \sg_{1-k}^{|\DD|} \, \big).
\label{sqrt}
\end{align}
Again, a straightforward induction argument shows that for any $j\in
\N$, $\sg_{1-j}^{|\DD|} \in O_{1-j}$, and thus $|\DD|\in \B_o$. The
result follows as above.
\end{proof}

In the next four lemmas, we emphasize the fact that only some of the results could be obtained using the trick \eqref{trick} with operator $J$. 

\begin{lemma}
\label{lemadxi}

(i) If $d$ is odd, then for any $P\in \B_e$, $\ncint P = 0$.

(ii) If $d$ is even, then for any $P\in \B_o$, $\ncint P = 0$.

(iii) For any pseudodifferential operator $P \in \Psi_1(\A)$ (see Appendix 5.1),

- when $d$ is odd, then $\ncint P =0$,

- when $d$ is even, then $\ncint P |\DD|^{-1} = 0$.
\end{lemma}

\begin{proof}
$(i)$ Since $\sg^P_{-d} \in E_{-d}$, $\sg^P_{-d}(x,\xi)=\sum_{i\in I}
\tfrac{\xi^{\b^i}}{\norm{\xi}_x^{2k_i}} \,h_i(x)$ where 
$|\b^i|$ are odd. The integration on the cosphere in \eqref{wres}
therefore vanishes.

$(ii)$ The same argument can be applied.

$(iii)$ Direct consequence of $(i)$ and $(ii)$.
\end{proof}

\begin{remark} Lemma \ref{lemadxi} (iii) entails for instance that $\ncint B|\DD|^{-(2k+1)}$ where $B$ is a polynomial in $\A$ and $\DD$ and $k\in \N$, always vanish in even dimension, while $\ncint B \DD^{-2k}$ always vanish in odd dimension. In other words, $\ncint B |D|^{-(d-q)} =0$ for any odd integer $q$.
\end{remark}

We shall now pay attention to the real or purely imaginary nature
(independently of the appearance of gamma matrices) of homogeneous
symbols of a given pseudodifferential operator. Let 
$$
\CC:=\set{P\in \Psi^p(M) \, : \,   \sg_{p-j}^{P} \in I_{j},\, \forall
j\in \N}
$$
where $I_k=I_e$ if $k$ is even and $I_{k}=I_o$ if $k$ is odd, with 
\begin{align*}
& I_e:=\set{f\in C^{\infty} \big(U\times \R^n,\M_m(\C) \big)
\,: \, f = \ga_{k_1}\cdots \ga_{k_q} \, h(x,\xi) \ , \ h \text{ real valued}}, \\
& I_o:=\set{f\in C^{\infty} \big(U\times \R^n,\M_m(\C) \big) \, : \, f = i \,\,\ga_{k_1}\cdots \ga_{k_q} \,h(x,\xi) \ , \ h \text{ real valued}}.
\end{align*}

\begin{lemma}
\label{Calg}
(i) $\CC$ is a sub-algebra of $\Psi(M)$.

\noindent (ii) If $P\in \CC$ is  hypo-elliptic then $P^{-1}\in \CC$.

\noindent (iii) $\DD^{k}\in \CC$ and  $|\DD|^k\in \CC$ for any $k\in
\Z$.
\end{lemma}

\begin{proof}
$(i)$ Consequence of (\ref{prodsymbol}).

\noindent $(ii)$ Consequence of (\ref{parametrix}).

\noindent $(iii)$ It is clear that $\DD \in \CC$ and the fact that
$|\DD|\in \CC$ is a consequence of (\ref{sqrt}).
\end{proof}

\begin{lemma}
\label{lemadk}
Let $k\in \N$ odd. Then any element $B$ of the polynomial algebra generated by $\A$ and $[D,\A]$ satisfies $\ncint B |\DD|^{-(d-k)} =\ncint B F|\DD|^{-(d-k)}=0$.
\end{lemma}

\begin{proof}
We may assume that $B$ is selfadjoint so $\ncint B \DD^{-(d-k)} \in
\R$.

By Lemma \ref{Calg}, $\sg_{-d}^{B \DD^{-(d-k)}} = \sg_0^{B} \sg_{-d}^{\DD^{-(d-k)}} 
\hspace{-0.1cm}
\in I_{k}$. Thus $\ncint A\DD^{-k}\in i \R$ and the result
follows. The case $\ncint B F|\DD|^{-(d-k)}$ is similar.
\end{proof}

We now look at the information given by the gamma matrices.

\begin{lemma}
For any one-form $A$, $\ncint A |\DD|^{-q}=0$, $q \in \N$ in either of the following cases:

- $d \neq 1 \mod 8$ and $d\neq 5 \mod 8$,

- ($d = 1 \mod 8$ or $d = 5 \mod 8$) and ($q$ is even or $q\geq \tfrac{d+3}{2}$).
\end{lemma}

\begin{proof} In the case $d \neq 1 \mod 8$ and $d\neq 5 \mod 8$, the result follows from the fact that $\eps=1$.

The case $d$ even and $q$ odd or $d$ odd and $q$ even is done by Lemma \ref{lemadxi} $(iii)$.

Suppose that $d$ is even and $q$ is even.
If $q=2k$, with a recurrence and the symbol product formula,
we see that $\sigma_{2k-j}^{D^{2k}}$ and all its derivatives are linear
combinations of terms of the form $f(x,\xi)\ox
\ga^{j_1}\cdots\ga^{j_i}$ where $i$ is even and less than $2j$ (with
the convention $\ga^{j_1}\cdots\ga^{j_i}=1$ if $i=0$). We call
$(P_{j})$ this property. The parametrix equation $ \DD^{2k}\DD^{-2k}=1$
entails that $\sigma_{-2k}^{\DD^{-2k}}=(\sigma_{2k}^{\DD^{2k}})^{-1}$ and
for any $j\geq 1$,
\begin{align*}
&\sigma_{-2k-j}^{\DD^{-2k}} = -
\sigma_{-2k}^{\DD^{-2k}}\big( \,\sum_{r=\max\{j-2k,0\}}^{j-1}
\sigma_{2k-(j-r)}^{\DD^{2k}}\,\sigma_{-2k-r}^{\DD^{-2k}}\nonumber\\
& \hspace{4cm}+ \sum_{1\leq |\a|\leq 2k}
\,\sum_{r=\max\{j-2k,0\}}^{j-|\a|}
\tfrac{(-i)^{|\a|}}{\a!} \,
\del^\a_{\xi}\sigma_{2k-(j-|\a|-r)}^{\DD^{2k}}\
\del^\a_{x}\sigma_{-2k-r}^{\DD^{-2k}}\,\big). 
\end{align*}
Note that $\sigma_{-2k}^{\DD^{-2k}}$ satisfies $(P_0)$. By recurrence, 
this formula shows  that
$\sigma_{-2k-j}^{\DD^{-2k}}$ satisfies $(P_{j})$ for any $j\in \N$. In
particular, $\sigma_{-d}^{\DD^{-2k}}$ satisfies $(P_{-2k+d})$ and the
result follows then from \eqref{sigmaA} and the product of an odd
number (different from the dimension) of gamma matrices is traceless.

Suppose now that $d$ is odd, $q$ is odd and $d\geq q$. In that situation, any odd number of gamma matrices $\ga^{i_1}\cdots \ga^{i_r}$ is traceless when $r<d$.

Using \eqref{prodsymbol} for the equation $|\DD|^{-q}|\DD|^{-q}=\DD^{-2q}$, we
check that $\sg_{-q}^{|\DD|^{-q}}(x,\xi)=\norm{\xi}_x^{-q}$ and for any $j\in
\N^*$,
\begin{align*}
\sg_{-q-j}^{|\DD|^{-q}}=&\tfrac{1}{2\norm{\xi}_x^{-q}}
\big(\,\sg_{-2q-j}^{\DD^{-2q}}-\sum_{0<k<j}\sg_{-q-j+k}^{|\DD|^{-q}} \,
\sg_{-q-k}^{|\DD|^{-q}} 
+ \sum_{0<|\a|\leq
j}\sum_{k=0}^{j-|\a|}i^{|\a|}\tfrac{(-1)^{|\a|}}{\a!} \, \del^\a_\xi
\sg_{-q-j+|\a|+k}^{|\DD|^{-q}} \del_{x}^\a \sg_{-q-k}^{|\DD|^{-q}} \, \big) .
\end{align*}
We saw that each $\sigma_{-2q-j}^{|\DD|^{-2q}}$ satisfies ($P_j$), that is to say, is a linear combination of terms of the form $f(x,\xi)\ox \ga^{j_1}\cdots \ga^{j_i}$ where $i$ is even and less than $2j$.
Again, a straightforward induction argument shows that for any $j\in
\N$, $\sg_{-q-j}^{|\DD|^{-q}}$ satisfies ($P_j$). In particular $\sg_{-d}(A|\DD|^{-q})$ is a linear combination of terms of the form $f(x,\xi)\ox \ga^{j_1}\cdots \ga^{j_r}$ where $r\leq 2(d-q)+1$ is odd. This yields the result.
\end{proof}

The fact that $\ncint A \,\DD^{-d+1}=0$, consequence of Lemmas
\ref{lemadxi} and \ref{lemadk} is also a consequence of the fact that
$\sigma_{-d}^{\DD^{-d+1}}(x_0,\xi)=0$:

\begin{lemma}
  \label{symbol-k-1}
For all $k\in \N^*$, we have
$\sigma_{k-1}^{\DD^k}(x_0,\xi)=\sigma_{-k-1}^{\DD^{-k}}(x_0,\xi)=0$.
\end{lemma}

\begin{proof}
We already know that $ \sigma_0^\DD(x_0,\xi)=0$, see \eqref{sigma0}. We
proceed by
recurrence, assuming $ \sigma_{k-1}^{\DD^k}(x_0,\xi)=0$ for
$k=1,\cdots,n$. Then $\sigma_n^{\DD^{n+1}}=\sigma_n^{\DD^n}
\sigma_0^\DD+\sigma_{n-1}^{\DD^n} \sigma_1^\DD -i\,
\partial_{\xi_k}\sigma_n^{\DD^n} \,\partial_{x^k} \sigma_1^\DD$, thus by
\eqref{sigma0} and \eqref{deriveedesigma1}, $\sigma_n^{\DD^{n+1}}
(x_0,\xi)=0$.

Since $\DD\DD^{-1}=1$ yields $\sigma_{-2}^{\DD^{-1}}
(x_0,\xi)=-\big(\sigma_{-1}^{\DD^{-1}}  \,\sigma_0^{\DD} \big)
(x_0,\xi)=0$, we  assume $\sigma_{-k-1}^{\DD^{-k}} (x_0,\xi)=0$ for
$k=1,\cdots n$. Then 
$\sigma_{-n-2}^{\DD^{-n-1}}=\sigma_{-n}^{\DD^{-n}}
\sigma_{-2}^{\DD^{-1}}+\sigma_{-n-1}^{\DD^{-n}} \sigma_{-1}^{\DD^{-1}} -i\,
\partial_{\xi_k}\sigma_{-n}^{\DD^{-n}} \,\partial_{x^k}
\sigma_{-1}^{\DD^{-1}}$. Using \eqref{deriveesigma-1} and recurrence
hypothesis, $\sigma_{-n-2}^{\DD^{-n-1}}(x_0,\xi)=0$.
\end{proof}

\begin{remark}

Regularity of $\zeta_X(s):=\Tr(\vert X \vert^{-s})$ at point 0 when
$X$ is an elliptic selfadjoint differential operator of order one
(see \cite{Gilkey}){\rm:} 
\end{remark}
One checks that   
$\zeta_X(s)=\tfrac{1}{\Gamma(s)}\int_0^{\infty} t^{s-1}
\Tr(e^{-t\vert X \vert})\, dt$ for $\Re (s)>d$. Because of the
asymptotic expansion 
\begin{align}
\label{exp}
\Tr(e^{-t \vert X \vert})=t^{-d}\,\sum_{n=0}^N t^n \, a_n[X]
+\mathcal{O}(t^{N+1-d})
\end{align}
and meromorphic extension to the whole complex plane,
$\underset{s=d-n}{\Res} \,\zeta_X(s)=\tfrac{a_n[X]}{\Gamma(d-n)}$. In
particular, $\zeta_X(s) =\Gamma(s)^{-1}\big( \tfrac{a_{d}[X]}{s}
+f(s)\big)$, where $f$ is holomorphic around $s=0$. By \eqref{Gamma}
we get
that $\zeta_X(s)$ is regular around zero and $\zeta_X(0)=a_{d}[X]$ if
$d$ is even and  $\zeta_X(0)=0$ if $d$ is odd.

\begin{corollary}
  \label{difference}
$\zeta_{\DD+A}(0)=\zeta_{\DD}(0)=0$ when $d=dim(M)$ is odd.

When $d$ is even, $\zeta_{\DD+A}(0)-\zeta_{\DD}(0)=\sum_{k=1}^{d/2}
\tfrac{1}{2k} \ncint (A\,\DD^{-1})^{2k}$.
\end{corollary}

\begin{proof}
The result follows from \eqref{constant} and Lemma
\ref{componentofzeta}.
\end{proof}

A proof of \eqref{constant} also follows from
$\sigma^{\log(1+A\DD^{-1})} \sim \sum_{k=1}^{\infty} \tfrac{(-1)^k}{k}
\, \sigma^{{(A\DD^{-1})}^k}$ with $\log (X):=\tfrac{\pa\,}{\pa
z}_{\vert_{z=0}} X^z$, so $Wres \big(\log(1+A\DD^{-1})
\big)=\sum_{k=1}^d \tfrac{(-1)^k}{k} \, Wres \big({{A\DD^{-1})}}^k
\big)$ since ${(A\DD^{-1})}^k$ has zero Wodzicki residue if $k>d$ and
moreover $\zeta_{\DD+A}(0)=-Wres \big( \log(\DD+A) \big)$. Actually, the
important point is that $\det(X):=e^{Wres \big(\log(X) \big)}$ is multiplicative (see
\cite{LP}.) Moreover, such determinant is different from the $zeta$-determinant $e^{-\zeta_X'(0)}$ used for instance by Hawking \cite{Hawking} in his regularization via the partition function which suffers from conformal anomalies.

\vspace{0.3cm}
The fact that in the asymptotic expansion of the heat kernel
\eqref{exp}, the term $a_2[\DD+A]$ depends only on the scalar
curvature, so independent of $A$ is reflected in 

\begin{lemma}
   \label{dim2}
In any spectral triple of dimension 2 (commutative or not) with
vanishing tadpoles of order zero (i.e. \eqref{equ} is satisfied),
$\zeta_{\DD+A}(0)=\zeta_{\DD}(0)$ for any one-form $A$.
\end{lemma}

\begin{proof}
Let $a_1,\,a_2,\,b_1,\,b_2\in \A$. Then, with $A_1=a_1[\DD,b_1]$,
 
\centerline{$\ncint A_1 \, \DD^{-1}\, a_2[\DD,b_2]\,\DD^{-1}=\ncint A_1
[\DD^{-1},a_2] [\DD,b_2]D^{-1}+ \ncint A_1 a_2 \DD^{-1}[\DD,b_2] \DD^{-1}$. }

The first term is zero since the integrand is in $OP^{-3}$, while the second term is equal to $\ncint
\big(a_1\a(b_1a_2)-a_1b_1\a(a_2)\big)\big(\a(b_2)-b_2\big)$,  so is
zero using $\a(x)\a(y)=\a(xy)$, $\ncint xy=\ncint x \a(y)$ by \eqref{equ} and the fact that $\ncint$ is a trace. 
Thus $\ncint
\big(A \DD^{-1}\big)^2=0$ and Corollary \ref{difference} yields the result.
\end{proof}

Note that $\zeta_{\DD+A}(0)-\zeta_{\DD}(0)$ is usually non zero:
consider for instance the flat 4-torus and as a generic selfadjoint
one-form $A$, take 
$$
A:= \phi \in[0,2\pi[^4 \,\mapsto -i\gamma^{\a}\,{\sum}_{l\in \Z^4}
\,a_{\a, l}\,e^{\,i\, l^k \phi_k},
$$
where $a_{\a, l}$ is in the Schwartz space $\SS(Z^4)$ and
$a_{\a,l}=-\overline{a_{\a,-l}}$. We have by \cite[Lemma 6.12]{MCC}
(with $c=\tfrac{8 \pi^2}{3}$,  $\vert l\vert^2={\sum}_k {l^k}^2$ and
$\Th=0$)
\begin{align*}
\zeta_{\DD+A}(0)-\zeta_{\DD}(0)=\ncint (A\DD^{-1})^2=c\,\sum_{l\in
\Z^4} a_{\a_1,l}\, a_{\a_2,-l}\, (l^{\a_1}l^{\a_2} -\delta^{\a_1\a_2} \vert l \vert^2)
\end{align*}
since $\ncint (A \DD^{-1})^4=0$.

This last  equality suggests that Lemma \ref{dim2} can be extended:

\begin{prop}
  \label{AD-1max} 
For any one-form $A$,  $\ncint \,(A\DD^{-1})^{d}=0$ if $d=dim(M)$.
\end{prop}

\begin{proof}
As in the proof of Lemma \ref{dim2}, $\DD^{-1}$ commutes with the
element in the algebra as the integrand is in $OP^{-d}$. So for a
family of $a_i,b_i \in \A$ and using $a:=\prod_{i=1}^d a_i$,
$$
\ncint \prod_{i=1}^d \big(a_i[\DD,b_i]\,D^{-1}\big)=\ncint
\big(\prod_{i=1}^d a_i \big) \prod_{i=1}^d \big(
[\DD,b_i]\,\DD^{-1}\big)=\ncint a \prod_{i=1}^d \big(\a(b_i)-b_i \big).
$$
We obtain, since $\a(b_i)-b_i \in OP^{-1}$, 
$$
\sigma_{-d}^{a\prod_{i=1}^d \a(b_i)-b_i} =  a\,\prod_{i=1}^d
\sigma_{-1}^{\a(b_i)-b_i} =  a\,\prod_{i=1}^d \sigma_{-1}^{\a(b_i)}
.$$
Moreover,  $\sigma_{-1}^{\DD b_i\DD^{-1}}(x_0,\xi)=0$: we already know by
Lemma \ref{symbol-k-1} that $\sigma_{-2}^{\DD^{-1}}(x_0,\xi)=0$, by
\eqref{sigma-1} that $\del_{x^k} \sigma_{-1}^{\DD^{-1}}(x_0,\xi)=0$ for
all $k$,  and $\sigma_0^{\DD b_i}(x_0,\xi) =b_i(x_0) \,
\sigma_0^\DD(x_0,\xi)=0$ giving the claim and the result. 
\end{proof}   

\vspace{0.2cm}
This proposition does not survive in noncommutative spectral triples,  see for instance \cite[Table 1]{MC}.

\medskip

Note that for a one-form $A$, $\ncint A^d\,\DD^{-d} \neq \ncint
(A\,D^{-1})^{-d}=0$: in dimension $d=2$, as in \eqref{d=2},
$$
\ncint A^2\,\DD^{-2}=-2c_d\, \Tr(\ga^k \ga^l)\int_M\,a_{k} a_{l}\, \nu_g.
$$

It is known (see \cite[Proposition 1.153]{ConnesMarcolli}) that the
$d-2$ term (for $d=4$) in the spectral action expansion $\ncint|\DD +
A|^{-2}$ is independent of the perturbation $A$. This is why the
Einstein--Hilbert action $S(\DD)=\ncint |\DD|^{-d+2}=-c \int_M \tau \sqrt{g} \, dx$ 
(see \cite[Theorem 11.2]{Polaris})  is so fundamental. Here $\tau $ is the scalar curvature 
(positive on the sphere) and $c$ is a positive constant. 

We give here another proof of this result. 

\begin{lemma}
\label{scalarcurvature}
We have $\ncint |\DD + A|^{-d+2} = \ncint |\DD|^{-d+2} =-c\int_M \tau \sqrt{g} \, dx$ with $c=\tfrac{d-2}{24} \, \ncint \vert \DD \vert^{-d}$.
\end{lemma}

\begin{proof}
We get from \cite[Lemma 4.10 $(ii)$]{MCC} the following
equality, where $X:=A\DD+\DD A+A^2$:
$$
\ncint |\DD + A|^{-d+2}-\ncint |\DD|^{-d+2} =
\tfrac{(d-2)}{2}\,\big(\tfrac{d}{4}\ncint X^2 |\DD|^{-d-2} - \ncint X
|\DD|^{-d}\big).
$$
Since the tadpole terms vanish, we have $\ncint X |\DD|^{-d} = \ncint
A^2 |\DD|^{-d}$. Moreover, since mod $OP^1$, $X^2 = (A\DD)^2 + (\DD A)^2 + A\DD^2
A + \DD A^2\DD$, we get with $[\DD^2,A]\in OP^1$,
$$
\ncint X^2 |\DD|^{-d-2} = 2\ncint (A\DD)^2|\DD|^{-d-2} + 2 \ncint A^2 |\DD|^{-d} 
$$ 
which yields
$$
\ncint |\DD + A|^{-d+2}-\ncint |\DD|^{-d+2} =
\tfrac{d(d-2)}{4} \big(\ncint (A\DD)^2 |\DD|^{-d-2} -\tfrac{2-d}{d} \ncint A^2 |\DD|^{-d}\big).
$$
Thus, it is sufficient to check that
$$ 
\int_{S_{x_0}^*M} \Tr\big(\sigma_{-d}((A\DD)^2 |\DD|^{-d-2})(x_0,\xi)\big)\, d\xi= \tfrac{2-d}{d}\int_{S_{x_0}^*M} \Tr\big(\sigma_{-d}(A^2 |\DD|^{-d})\,(x_0,\xi)\big) d\xi.
$$

A straightforward computation yields, with $A=:-i a_\mu \ga^\mu$, and $\sg_1^\DD(x_0,\xi) = \ga^\mu \xi_\mu$,
\begin{align*}
&\int_{S^*_{x_0} M}\sg_{-d}((A\DD)^2 |\DD|^{-d-2})(x_0,\xi)\, d\xi =-
\tfrac{1}{d}\, a_\mu a_\tau \Tr(\ga^{\mu} \ga^\nu \ga^\tau \ga_{\nu})
\,\text{Vol}(S^{d-1})\, ,\\
&\int_{S^*_{x_0} M} \sg_{-d}(A^2 |\DD|^{-d})(x_0,\xi)\, d\xi = -a_\mu a_\tau
\Tr(\ga^\mu \ga^\tau) \,\text{Vol}(S^{d-1})\, .
\end{align*}
Now, $\ncint |\DD + A|^{-d+2} = \ncint |\DD|^{-d+2}$ follows from the equality $\Tr(\ga^{\mu} \ga^\nu
\ga^\tau \ga_{\nu})= (2-d) \Tr(\ga^\mu \ga^\tau)$.
The constant $c$ is given in \cite[Theorem 11.2 and normalization (11.2)]{Polaris}.
\end{proof}
\begin{remark}
In \cite[Definition 1.143]{ConnesMarcolli}, the above result justifies the definition of a scalar curvature for $(\A,\H,\DD)$ as $\mathcal{R}(a):=\ncint a \vert \DD \vert^{-d+2}$ for $a \in \A$. This map is of course a trace on $\A$ for a commutative geometry. But for the triple associated to $SU_q(2)$, this not a trace since (see \cite{MC}):
\begin{align*}
\mathcal{R}(aa^*)=\ncint aa^* \, \vert \DD\vert^{-1}= \tfrac{-q^4+6q^2+3}{2{(1-q^2)}^2}  \quad \text{while } \quad
\mathcal{R}(a^*a)=\ncint a^*a \,\vert \DD\vert^{-1}= \tfrac{3q^4+6q^2-1}{2{(1-q^2)}^2} \,.
\end{align*}
\end{remark}

\section{Appendix}

\subsection{Pseudodifferential operators}

\begin{definition}
\label{defpseudo} Let us define $\DD(\A)$ as the polynomial algebra
generated by $\A$, $J\A J^{-1}$, $\DD$ and $|\DD|$.

A pseudodifferential operator is an operator $T$ such that there
exists $d\in \Z$ such that for any $N\in \N$, there exist $p\in \N_0$, $P\in \DD(\A)$ and $R\in
OP^{-N}$ ($p$, $P$ and $R$ may depend on $N$) such that $P\,\DD^{-2p}\in OP^d$ and
$$
T=P\,\DD^{-2p}+R\, .
$$
Define $\Psi(\A)$ as the set of pseudodifferential operators and
$\Psi(\A)^k:=\Psi(\A)\cap OP^k$.
\end{definition}
Note that the notion of pseudodifferential operator is modified as
$\Psi(\A)$ now includes $J \A J^{-1}$, see \cite{MCC}.

When $A$ is a one-form, $A$ and $JAJ^{-1}$ are in $\DD(\A)$
and moreover $\DD(\A)\subseteq \cup_{k \in \N_0} OP^k$. Since $|\DD|\in \DD(\A)$
by construction and $P_0$ is a pseudodifferential operator, for any
$k\in \Z$, $|D|^{k}$ is a pseudodifferential operator (in $OP^{k}$.)
Let us remark also that $\DD(\A)\subseteq\Psi(\A) \subseteq
\cup_{k\in \Z} OP^{k}$.

The set of all pseudodifferential operators $\Psi(\A)$ is an algebra. 
We denote $\Psi_1(\A)$ the subalgebra of $\Psi(\A)$ defined the same way as $\Psi(\A)$, replacing $\DD(A)$ by the polynomial algebra generated by $\DD, \A$ and $J\A J^{-1}$.
This algebra is similar to the one defined in \cite{CC1}.

\subsection{Zeta functions and dimension spectrum}

For any operator $B$ and if $X$ is either $D$ or
$D_A$, we define
\begin{align*}
{\zeta}_X^B(s)&:= \Tr\big(B|X|^{-s}\big ),\\
\zeta_X(s)&:= \Tr \big(|X|^{-s}\big).
\end{align*}

\medskip

{\it The dimension spectrum} $Sd(\A,\H,\DD)$ of a
spectral triple has been defined in
\cite{Cgeom, CM}. It is extended here to pay attention
to the operator $J$ and to our definition of pseudodifferential
operator.

\begin{definition}
The spectrum dimension of the spectral triple is
the subset $Sd(\A,\H,\DD)$ of all poles of the functions
$\zeta_\DD^P := s\mapsto \Tr \big(P |\DD|^{-s}\big)$ where $P$ is any
pseudodifferential operator in $OP^0$. The spectral triple
$(\A,\H,\DD)$ is said to be simple when these poles are all simple.
\end{definition}

The following is part of folklore in noncommutative geometry, even if
sometimes it is unclear if there is an equality or an inclusion of
$Sp(M)$ in $\set{d-k \, : \, k \in \N}$.
\begin{prop}
Let $Sp(M)$ be the spectrum dimension of a commutative geometry of dimension $d$. Then $Sp(M)$ is simple and $Sp(M) = \set{d-k
\, : \, k \in \N}$.
\end{prop}

\begin{proof}
Let $a\in \A=C^{\infty}(M)$ such that its trace norm $\vert\vert a
\vert\vert_{\L^1}$ is non zero and for $k \in \N$, let $P_k:=a \vert
D \vert^{-k}$. Then $P_k \in OP^{-k} \subset OP^{0}$ and its
associated zeta-function has a pole at $d-k$:
\begin{align*}
\underset{s=d-k}{\Res} \,\zeta^P_\DD(s)&=\underset{s=0}{\Res}
\,\zeta_\DD^P(s+d-k)=\underset{s=0}{\Res} \,\Tr \big(a \vert
\DD\vert^{-k}\vert \DD\vert^{-(s+d-k)}\big)=\ncint a \vert \DD \vert ^{-d}\\
&=\int_M a(x) \int_{S_x^*M} \Tr\big(({\sigma_1^{\vert
\DD\vert}})^{-d}(x,\xi) \big) \vert d\xi \vert \, \vert dx \vert
=\int_M a(x) \int_{S_x^*M} \vert\vert \xi \vert \vert^{-d/2}  \vert
d\xi \vert \, \vert dx \vert \\
&=\int_M a(x) \,\nu_g(x) =\vert\vert a \vert\vert_{\L^1}\neq 0
\end{align*}
where $\nu_g$ is the Riemann density normalized on $g$-orthonormal
basis of $TM$.

Conversely, since $\Psi(\A)^0$ is contained in the algebra of all
pseudodifferential operators of order less or equal to 0, it is known
\cite{Guillemin,Wodzicki1,Wodzicki} that $Sp(M) \subset \set{d-k \, :
\, k \in \N}$.

The fact that all poles are simple is due to the fact that $\DD$
being differential and $M$ being without boundary,  $a'_k=0$, $\forall k\in \N^*$ in \eqref{heat}.
\end{proof}

\section*{Acknowledgments}

\hspace{\parindent}
We thank Jean-Marie Lescure, Thomas Sch\"ucker and Dmitri Vassilevich for helpful discussions.

\end{document}